\documentclass[a4paper]{article}

\usepackage[T1]{fontenc}
\usepackage{graphicx}
\usepackage[utf8]{inputenc}
\usepackage[top=2cm,bottom=2cm,left=3cm,right=3cm,marginparwidth=2cm]{geometry}
\usepackage{amsmath}
\usepackage{amssymb}
\usepackage{amsthm}
\usepackage{mathtools}
\usepackage{xparse}
\usepackage{enumerate}
\usepackage[inline]{enumitem}
\usepackage{todonotes}
\usepackage{tikz}
\usepackage{subcaption}
\usepackage{multirow}
\usepackage{pgfplots}
\usepackage{float}
\usepackage{utfsym}

\newtheorem{theorem}{Theorem}
\newtheorem{corollary}[theorem]{Corollary}

\DeclareDocumentCommand\tableYes{}{\usym{2713}}
\DeclareDocumentCommand\tableNo{}{\usym{2717}}

\DeclareDocumentCommand\R{}{\mathbb{R}}
\DeclareDocumentCommand\costOne{}{\ensuremath{C_1}}
\DeclareDocumentCommand\costTwo{}{\ensuremath{C_2}}

\DeclareDocumentCommand\actions{}{\ensuremath{\mathcal{A}}}
\DeclareDocumentCommand\actionsOne{}{\ensuremath{\mathcal{A}_1}}
\DeclareDocumentCommand\actionsTwo{}{\ensuremath{\mathcal{A}_2}}

\DeclareDocumentCommand\labelOneOpt{}{\ensuremath{\textcolor{blue}{A_1^{\text{opt}}}}}
\DeclareDocumentCommand\labelTwoOpt{}{\ensuremath{\textcolor{green!50!black}{A_2^{\text{opt}}}}}
\DeclareDocumentCommand\labelOneEqui{}{\ensuremath{\textcolor{red}{A_1^{\text{equi}}}}}
\DeclareDocumentCommand\labelTwoEqui{}{\ensuremath{\textcolor{orange}{A_2^{\text{equi}}}}}
\DeclareDocumentCommand\labelTwoEquiX{}{\ensuremath{\textcolor{purple}{A_2^{\text{equi}'}}}}

\usetikzlibrary{shapes.geometric}
\usetikzlibrary{positioning}

\usepackage[colorlinks=true, allcolors=blue]{hyperref}
\usepackage[capitalize,noabbrev]{cleveref}

\title{Exact Price of Anarchy for Weighted Congestion Games with Two Players}
\author{Joran van den Bosse \and Marc Uetz \and Matthias Walter}

\DeclareDocumentCommand\colorOne{m}{\textcolor{blue}{#1}}
\DeclareDocumentCommand\colorTwo{m}{\textcolor{green!60!black}{#1}}

\pgfplotsset{compat=1.15}

\begin{document}

\maketitle

\begin{abstract}
  This paper gives a complete analysis of worst-case equilibria for various versions of weighted congestion games with two players and affine cost functions.
  The results are exact price of anarchy bounds which are parametric in the weights of the two players, and establish exactly how the primitives of the game enter into the quality of equilibria.
  Interestingly, some of the worst-cases are attained when the players' weights only differ slightly.
  Our findings also show that sequential play improves the price of anarchy in all cases, however, this effect vanishes with an increasing difference in the players' weights.
  Methodologically, we obtain exact price of anarchy bounds by a duality based proof mechanism,
  based on a compact linear programming formulation that computes worst-case instances.
  This mechanism yields duality-based optimality certificates which can eventually be turned into purely algebraic proofs.
\end{abstract}

\section{Introduction}

This paper studies the quality of equilibria for games with two players that compete for a set of resources, when the cost (or congestion) on each resource is given by an affine function that depends on the total load of that resource. The games that we consider fall into a class of games known as weighted Rosenthal congestion games \cite{Rosenthal73,Milchtaich96}. 
As we address several variants of such games, before discussing our actual contribution, let us first define these different settings.

\paragraph{Problem definition.} There are two players denoted $i=1,2$, and the possible \emph{actions} of player $i$ are given by admissible subsets $\mathcal{A}_i \subseteq 2^R$ of a finite set of resources $R$. This includes as a special case so-called \emph{atomic network routing games}, where the resources are edges $E$
in some directed graph $G=(V,E)$, and each player wants to establish a path between two source and target vertices $s_i,t_i\in V(G)$, so that the admissible actions of player~$i$ are precisely the edge sets of the directed  $(s_i,t_i)$-paths.  A congestion game is called \emph{symmetric} if $\mathcal{A}_1 = \mathcal{A}_2$. For network routing games, that means that all players have the same source and the same target.

We address two different game theoretic settings. In \emph{simultaneous} games, both players choose their actions simultaneously, and the admissible actions $\mathcal{A}_i$ of a player~$i$ are exactly the player's strategies in the corresponding strategic form bimatrix game. 
We also consider \emph{sequential} extensive form games where the players choose actions after each other, w.l.o.g.\ first player 1, then player 2. Here, the first player's strategies are the actions $\mathcal{A}_1$, and a strategy for the second player consists of one action from $\mathcal{A}_2$ for \emph{each} possible strategy ($=$ action) of the first player. In both cases, the game results in an outcome where both players have chosen one action, denoted \emph{action profile} $A=(A_1,A_2)$, where $A_i\in\mathcal{A}_i$.

The players have to pay for each resource $r\in R$ that they choose, and the cost of a resource~$r\in R$ depends on its load, which again depends on the set of players using it.  In the unweighted version of the problem, the load $x_r$ of a resource~$r$ equals the \emph{number} of players using it, and the game is a potential game \cite{Rosenthal73,MondererS1996}. In the weighted version of the problem, each player~$i$ has a \emph{player-specific weight $w_i$}, and the load $x_r$ of a resource~$r\in R$ equals the total weight of the players that have chosen it. So if 
$A=(A_1,A_2)$ is an action profile,  the load of a resource~$r\in R$ equals
$x_r(A)=\sum_{i| r\in A_i} w_i$.
When $w_1=w_2$, this corresponds to the unweighted version (possibly after scaling).

Throughout the paper, we assume that the cost function for each resource~$r$ is an affine function in its load $x_r$, with non-negative coefficients.
That is, for each resource $r \in R$, there are non-negative reals $\alpha_r \geq 0$ and $\beta_r \geq 0$, and with $x_r$ being the total load of resource~$r$, the cost of resource $r$ equals $\alpha_r + \beta_r x_r$.

For weighted congestion games, two different cost functions appear in the literature, under different names. We here follow the nomenclature as used e.g.\ in~\cite{HarksK16}.
The first class of cost functions is  \emph{uniform costs}~\cite{FotakisKS05,GairingMT06,IeongMNSS05,Milchtaich06,PanagopoulouS07,HarksK16}, where the cost of a resource can be thought of as a delay or latency (as e.g.\ in traffic) which is the same for all players choosing that resource, so that each player pays the costs for the loads of all chosen resources. That means that
an action profile $A=(A_1,A_2)$ yields as cost for player~$i$
\begin{equation}
  C_i^{\text{uni}}(A) \coloneqq \sum_{r \in A_i} (\alpha_r + \beta_r \sum_{j:\, r \in A_j} w_j)\,. \label{eq_cost_uni}
\end{equation}
The second class of cost functions is \emph{proportional costs}~\cite{GoemansMV05,HarksK16}, where the cost of a resource can be thought of as as per-unit cost, and each player pays proportionally to the load that the player imposes on a resource. Then, 
an action profile $A=(A_1,A_2)$ yields as cost for player~$i$
\begin{equation}
  C_i^{\text{prop}}(A) \coloneqq w_i \sum\limits_{r \in A_i} (\alpha_r + \beta_r \sum\limits_{j :\, r \in A_j} w_j). \label{eq_cost_prop}
\end{equation}
In the following we also use $C_i(\,\cdot\,)$ to denote either of the two cost functions.
Note that the cost functions agree in the (unweighted) case when $w_1 = w_2 = 1$.



%
\paragraph{Equilibria.}
For simultaneous games, a strategy profile $(\labelOneEqui, \labelTwoEqui)$ is a (pure) Nash equilibrium~\cite{Nash1950,Nash1951} if none of the players can improve unilaterally, i.e.,
\begin{subequations}
  \label{eq_lp_nash}
  \begin{align}
    \costOne(\labelOneEqui,\labelTwoEqui) &\leq \costOne(A_1,\labelTwoEqui) &&\forall A_1 \in \actionsOne \label{eq_lp_nash_change_one} \\
    \costTwo(\labelOneEqui,\labelTwoEqui) &\leq \costTwo(\labelOneEqui,A_2) &&\forall A_2 \in \actionsTwo \label{eq_lp_nash_change_two}
    \end{align}
\end{subequations}
Since we assume that the cost functions per resource are affine, the existence of pure Nash equilibria is guaranteed~\cite{HarksKlimmMOR2012}.
Note that this need not be true for arbitrary non-decreasing cost functions, not even for weighted symmetric two-player games \cite{FotakisKS05}. For unweighted congestion games, pure Nash equilibria always exist~\cite{Rosenthal73}.

For sequential games, let us assume w.l.o.g.\ that the players choose their actions in the order 1, and then 2. The first player then has strategy set $\actionsOne$.
A strategy of the second player can be written as a tuple of length $|\actionsOne|$, specifying the response of player~2 to any possible strategy of player~1.
Since we consider full information games, in equilibrium both players can choose a (pure) strategy that maximixes the player's cost.
For player~2 that means, if $A_1$ is chosen by player~1, to choose any 
$\labelTwoEqui$
so that 
\[
  C_2(A_1,\labelTwoEqui) \leq C_2(A_1,A_2) \text{ for all } A_2 \in \actionsTwo \, .
\]
If we denote by $\labelTwoEqui(A_1)$, $A_1\in\actionsOne$, such an equilibrium strategy for player~2, player~1 chooses any strategy $\labelOneEqui$ so as to minimize her cost, that is,  
\[
  C_1(\labelOneEqui,\labelTwoEqui(\labelOneEqui)) \leq C_2(A_1,\labelTwoEqui(A_1)) \text{ for all } A_1 \in \actionsOne \, .
\]
All pairs of strategies $(\labelOneEqui,\labelTwoEqui(\labelOneEqui))$ that fulfill both conditions are precisely the (pure) subgame-perfect equilibria~\cite{Selten1973} for the sequential, extensive form game.
For the full information games considered here, (pure) subgame perfect strategies exist and can easily be computed by the ``procedure'' as sketched above, known as backward induction~\cite{Peters2008}.
A subgame perfect equilibrium yields as outcome an action profile  $(\labelOneEqui, \labelTwoEqui)$.
It is well known, even for network routing games, that subgame perfect outcomes need not be a Nash equilibrium in the corresponding simultaneous game~\cite{CorreaJKU19}, which is one of the difficulties in analyzing subgame-perfect equilibria.

\paragraph{Price of anarchy.} It is well known that if players choose their actions $A=(A_1,A_2)$ selfishly while only considering their own cost functions $C_i(A)$, there may be outcomes that are stable with respect to either Nash or subgame-perfect equilibrium, but that solution need not minimize the cost of both players together $C(A):=C_1(A)+C_2(A)$.  The price of anarchy~\cite{Koutsoupias2009} measures these negative effects of selfish behaviour. It is defined as the maximum cost of an equilibrium outcome, relative to the cost of the so-called \emph{social optimum} which is an action profile that minimizes total cost. Formally, for a given instance $I$ of a weighted congestion game, if we define $\mathcal{A}^{\text{equi}}(I)$ as the set of action profiles that may result as outcome from some equilibrium, and $A^{\text{opt}}(I)$ as a social optimum, so any profile that minimizes total costs $C(\,\cdot\,)$, the price of anarchy is 
\[
  \textup{PoA}(I) \coloneqq \max_{A \in \mathcal{A}^{\text{equi}}(I)}  \frac{C(A)}{C(A^{\text{opt}}(I))} \, .
\]
For sequential games, the maximum is also taken over all possible orders of the players.
The price of anarchy for a class of games is the supremum of $\textup{PoA}(I)$ over all instances $I$ of that class. For sequential games, the price of anarchy has also been called the \emph{sequential price of anarchy} \cite{PaeLeme2012}; it  has been analyzed in several settings, e.g.\ \cite{PaeLeme2012,Bilo2013,DeJong2014,JongU19}.



\DeclareDocumentCommand\gapClosed{m}{\textcolor{green!50!black}{#1}}
\DeclareDocumentCommand\gapOpen{m}{\textcolor{red}{#1}}

\begin{table}[H]
  \caption{%
    Known results and improvements for lower bounds (lb) and upper bounds (ub) on the price of anarchy for weighted \emph{simultaneous} congestion games with two players.
    The first columns indicate restriction (\tableYes) to network or symmetric games or no restriction (\tableNo), as well as the cost function. In the first column, ``\tableNo, \tableYes\,'' indicates that the upper bound holds in general, while the lower bound is attained even for network routing games.
  }
  \label{tab_literature_sim}
  \begin{center}
    \renewcommand{\arraystretch}{1.1}
    \setlength\tabcolsep{1ex}
    \begin{tabular}{c|c|c|c|c|c|c|l}
      \textbf{Net}
      & \textbf{Sym}
      & \textbf{Cost}
      & \textbf{Old lb}
      & \textbf{New lb}
      & \textbf{New ub}
      & \textbf{Old ub}
      & \textbf{Result} \\ \hline
       \tableNo, \tableYes & \tableYes &        & \gapClosed{$1.6$}~\cite{CorreaJKU19} &  &  & \gapClosed{$1.6$}~\cite{ChristodoulouK05} \\ \hline
       \tableNo, \tableYes & \tableYes & uni    & $1.6$~\cite{CorreaJKU19} & \gapClosed{$2$} & \gapClosed{$2$} & & Thm.~\ref{thm_sym_uni_network}, Cor.~\ref{thm_sym_uni_all_weights} \\ \hline
       \tableNo, \tableYes & \tableYes & prop   & $1.6$~\cite{CorreaJKU19} & \gapClosed{$\approx 1.6096$} & \gapClosed{$\approx 1.6096$} & $\approx 2.618$~\cite{AwerbuchAE05} & Cor.~\ref{thm_sym_prop_all_weights}\\ \hline
       \tableNo, \tableYes & \tableNo &                    & \gapClosed{$2$}~[folklore]  
       &  &  & \gapClosed{$2$}~\cite{ChristodoulouK05}  \\
       \hline
       \tableNo, \tableYes & \tableNo & uni                & $2$~[folklore] &  \gapClosed{$\approx 2.155$} & \gapClosed{$\approx 2.155$} & & Cor.~\ref{thm_sim_uni_all_weights} \\ \hline
       \tableNo, \tableYes & \tableNo & prop               & $2$~[folklore] & \gapClosed{$\approx 2.0411$} & \gapClosed{$\approx 2.0411$} & $\approx 2.618$~\cite{AwerbuchAE05} & Cor.~\ref{thm_sim_prop_all_weights} 
    \end{tabular}
  \end{center}
\end{table}

\paragraph{Contribution and known results.}

This paper gives the exact price of anarchy results for several classes of \emph{weighted} two-player congestion games with respect to Nash equilibria for simultaneous games, and  subgame-perfect equilibria for sequential games.
We further distinguish between arbitrary congestion games and network routing games, uniform and proportional cost functions, as well as symmetric and asymmetric games. This gives rise to 16 different classes of two-player games. An overview of previously known bounds as well as new ones presented in this paper can be found in \cref{tab_literature_sim,tab_literature_seq}. The table also contains the previously known results for the special case of unweighted congestion games (if known). We not only give the worst-case bounds, but also prove the exact price of anarchy bounds parametric in the weight ratio $w_1/w_2$ of the two players, which provides additional insight. \cref{fig_plot} depicts these findings.

Our results close a gap in the existing literature on the analysis of equilibria for affine congestion games.
Let us therefore briefly discuss what is known. 

The price of anarchy for \emph{unweighted} affine congestion games with an {arbitrary number of players} is known to be bounded from above by 5/2 \cite{AwerbuchAE05,ChristodoulouK05}. This bound is tight for $n\ge 3$ players \cite{ChristodoulouK05}. For unweighted two-player games, the tight bound is 2~\cite{ChristodoulouK05}. The 5/2 upper bound improves to $(5n-2)/(2n+1)$ for \emph{symmetric} games with $n$ players~\cite{ChristodoulouK05}, which is tight even in the class of network routing games~\cite{CorreaJKU19}.
For non-symmetric, unweighted and affine congestion games, the bound 5/2 is tight (for $n\to\infty$) even for \emph{singleton} congestion games, where  each player chooses a single resource only, i.e., $|A_i|=1$ for all $A_i\in\mathcal{A}_i$~\cite{Caragiannis2010}.
For unweighted and affine singleton congestion games that are symmetric, so when each player can choose every single resource, the price of anarchy equals $4/3$~\cite{Luecking2008}.
This singleton model can also be interpreted as a network routing model with parallel source-sink arcs, and is also dubbed the \emph{parallel link} model, see e.g.~\cite{Roughgarden2003}.

For \emph{weighted} affine congestion games, the price of anarchy with an {arbitrary number of players} is slightly larger than 5/2; it equals $1+\phi$ with $\phi=(1+\sqrt{5})/2\approx 1.618$ being the golden ratio~\cite{AwerbuchAE05}. This is again tight for $n\ge 3$ players~\cite{DBLP:conf/waoa/Bilo12}. Here, we are not aware of improved bounds for special cases or with two players.

\begin{table}[H]
  \caption{%
    Known results and improvements for lower bounds and upper bounds on the price of anarchy for weighted \emph{sequential} congestion games with two players.
    The meaning of the columns is the same as in \cref{tab_literature_sim}.
  }
  \label{tab_literature_seq}
  \begin{center}
    \renewcommand{\arraystretch}{1.1}
    \setlength\tabcolsep{1ex}
    \begin{tabular}{c|c|c|c|c|c|c|l}
      \textbf{Net}
      & \textbf{Sym}
      & \textbf{Cost}
      & \textbf{Old lb}
      & \textbf{New lb}
      & \textbf{New ub}
      & \textbf{Old ub}
      & \textbf{Result} \\ \hline
      \tableNo, \tableYes & \tableYes &      & \gapClosed{$1.4$}~\cite{CorreaJKU19} &  & & \gapClosed{$1.4$}~\cite{CorreaJKU19} \\ \hline
      \tableNo, \tableYes & \tableYes & uni  & \gapOpen{$1.4$}~\cite{CorreaJKU19} &  & \gapOpen{$2$} & & Cor.~\ref{thm_seq_uni_all_weights} \\ \hline
      \tableNo, \tableYes & \tableYes & prop & \gapOpen{$1.4$}~\cite{CorreaJKU19} &  & \gapOpen{$1.5$} & & Cor.~\ref{thm_seq_prop_all_weights} \\ \hline
      \tableNo & \tableNo &                              & \gapClosed{$1.5$}~\cite{JongU19} &  &  & \gapClosed{$1.5$}~\cite{JongU19} \\ \hline
      \tableYes & \tableNo &                  & 1 & \gapClosed{$1.5$} &  & \gapClosed{$1.5$}~\cite{JongU19} \\ \hline
      \tableNo, \tableYes & \tableNo & uni              & $1$ & \gapClosed{$2$} & \gapClosed{$2$} & & Cor.~\ref{thm_seq_uni_all_weights} \\ \hline
      \tableNo, \tableYes & \tableNo & prop             & $1$ & \gapClosed{$1.5$} & \gapClosed{$1.5$} & & Cor.~\ref{thm_seq_prop_all_weights}
    \end{tabular}
  \end{center}
\end{table}

On the other hand, for the \emph{sequential} version of \emph{unweighted}, affine congestion games, it is known that the (sequential) price of anarchy equals $1.5$ for two players~\cite{JongU19},  it equals $1039/488\approx 2.13$ for $n=3$ players~\cite{JongU19}, and for $n=4$ players it equals $28679925/10823887 \approx 2.65$~\cite{Bosse2021}. Moreover, for symmetric network routing games and two players, it equals $1.4$~\cite{CorreaJKU19}, while it can be as large as $\Omega(\sqrt{n})$ for an arbitrary number of players~\cite{CorreaJKU19}.

The above summary suggests that the potential loss of efficiency caused by selfish players in  affine congestion games is by now pretty well understood. Yet we believe that the two-player case is a fundamental base case that deserves special attention, and it does not seem to be well understood when the two players are not identical, i.e., the weighted case. 
Also, our results exhibit how the primitives of the game, such as symmetry of strategies, proportionality of costs, or sequentiality, have effects on the price of anarchy, which we believe are interesting and not obvious.

As to the technical approach and contribution of this paper, all our results take as starting point a linear programming based idea to compute the cost functions of worst-case instances that extends the one used for unweighted congestion games in~\cite{JongU19}. The linear program computes the parameters of all resource cost functions for the finite worst case instance, which can be done because it can be proved that such finite worst-case instances 
must exist~\cite{DeJong2014,JongU19}. However for weighted games this can only be done for any \emph{fixed} pair of weights $w_1,w_2$. Doing this for several weight ratios $w_1/w_2$, however, one can produce ``educated guesses'' for the dependence of the cost parameters on the weights, and the same can be done for the corresponding optimal dual solution. The so derived expressions for primal and dual solutions can finally be turned into algebraic proofs for a matching upper bound on the price of anarchy. This primal-dual procedure is somewhat mechanical, yet it can  turn out to be quite tedious. It should be noted that this approach is reminiscent of the primal-dual technique as suggested in~\cite{DBLP:conf/waoa/Bilo12}, however here we work with a different linear programming formulation.

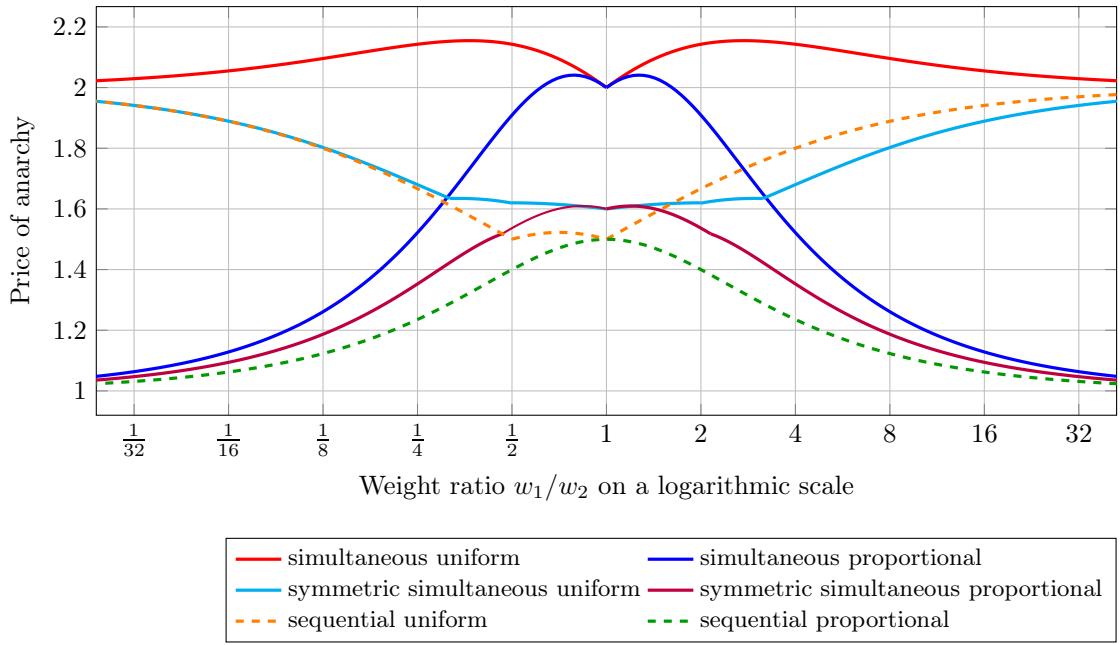
\begin{figure}[htbp]
  \begin{center}
    \edef\plotSamples{100} 
    \begin{tikzpicture}
      \begin{semilogxaxis}[
        xlabel={Weight ratio $w_1 / w_2$ on a logarithmic scale},
        ylabel={Price of anarchy},
        grid=both,
        xmin=0.03125,
        xmax=32,
        xtick={0.03125,0.0625,0.125,0.25,0.5,1,2,4,8,16,32},
        xticklabels={$\frac{1}{32}$,$\frac{1}{16}$,$\frac{1}{8}$,$\frac{1}{4}$,$\frac{1}{2}$,1,2,4,8,16,32},
        ytick={1.0,1.2,1.4,1.6,1.8,2.0,2.2,2.4,2.6},
        width=\textwidth,
        height=70mm,
        enlarge x limits=0.04,
        enlarge y limits=0.10,
        legend style={at={(1.0,-0.3)},anchor=north east,font=\small},
        legend columns=2,
        legend cell align={left},
        ]

        \addplot[red,very thick, domain=1:50, samples=\plotSamples] {1 + (2*x+x^2) / (x^2+x+1)};
        \addlegendentry{simultaneous uniform}

        \addplot[blue,very thick, domain=1:50, samples=\plotSamples] {1 + x * (x+1+x) / (x^3+1+x*1)};
        \addlegendentry{simultaneous proportional}
        
        \addplot[cyan,very thick, domain=1:2, samples=\plotSamples] {(3*x^3 + 9*x^2 + 9*x+3) / (2*x^3 + 5*x^2 + 6*x + 2)};
        \addlegendentry{symmetric simultaneous uniform}

        \addplot[purple,very thick, domain=2.1479:50, samples=\plotSamples] {(2*x^4 + 4*x^3 + 4*x^2 + 2*x) / (2*x^4+x^3+2*x^2+3*x+1)};
        \addlegendentry{symmetric simultaneous proportional}

        \addplot[orange,dashed,very thick, domain=1:50, samples=\plotSamples] {1 + x/(1+x)};
        \addlegendentry{sequential uniform}

        \addplot[green!60!black,dashed,very thick, domain=0.02:50, samples=\plotSamples] {1 + x/(x^2+1)};
        \addlegendentry{sequential proportional}
        
        \addplot[red,very thick, domain=0.02:1, samples=\plotSamples] {1 + (2*x+1) / (x^2+x+1)};

        \addplot[blue,very thick, domain=0.02:1, samples=\plotSamples] {1 + x * (x+1+1) / (x^3+1+x*x)};

        \addplot[cyan,very thick, domain=0.5:1, samples=\plotSamples] {(3*x^3 + 9*x^2 + 9*x+3) / (2*x^3 + 6*x^2 + 5*x + 2)};
        \addplot[cyan,very thick, domain=3.1527576:50, samples=\plotSamples] {(2*x^2 + 2*x+2) / ((x+1)^2)};
        \addplot[cyan,very thick, domain=0.02:0.3171826, samples=\plotSamples] {(2*x^2 + 2*x+2) / ((x+1)^2)};
        \addplot[cyan,very thick, domain=2:3.1527576, samples=\plotSamples] {
        (2*x^3 + 5*x^2 + 5*x + 2 + 3*x^2 + x^3) / (2*x^3 + 4*x^2 +4*x + 2)};
        \addplot[cyan,very thick, domain=0.3171826:0.5, samples=\plotSamples] {
        (2*x^3 + 5*x^2 + 5*x + 2 + 3*x + 1) / (2*x^3 + 4*x^2 +4*x + 2)};
        
        \addplot[purple,very thick, domain=0.02:0.4655, samples=\plotSamples] {(2*x^3 + 4*x^2 + 4*x + 2) / (x^4+3*x^3+2*x^2+x+2)};
        \addplot[purple,very thick, domain=1:2.1479, samples=\plotSamples] {(2*x^4+6*x^3+8*x^2+6*x+2) / (2*x^4+3*x^3+4*x^2+4*x+2)};
        \addplot[purple,thick, domain=0.4655:1, samples=\plotSamples] {(2*x^4+6*x^3+8*x^2+6*x+2) / (2*x^4+4*x^3+4*x^2+3*x+2)};

        \addplot[orange,dashed,very thick, domain=0.5:1, samples=\plotSamples] {1 + (2*x)/(2*x^2 + x + 1)};
        \addplot[orange,dashed,very thick, domain=0.02:0.5, samples=\plotSamples] {1 + 1/(2*x+1)};
      \end{semilogxaxis}
    \end{tikzpicture}
  \end{center}
  \caption{Price of anarchy for different types of weighted affine two-player congestion games depending on the weight ratio $w_1 / w_2$.}
  \label{fig_plot}
\end{figure}

\paragraph{Outline.}
We first present our results on the price of anarchy for simultaneous, symmetric simultaneous and sequential games in \cref{sec_sim,sec_sym,sec_seq}.
Our methodology for obtaining these results is explained in \cref{sec_lp} and a few concluding remarks are made in \cref{sec_conclusion}.

\section{Results for simultaneous games}
\label{sec_sim}

\begin{figure}[H]
  \begin{center}
    \begin{tikzpicture}[
      node/.style={draw,thick,circle,inner sep=0mm,minimum size=4mm},
      arc/.style={draw,->,thick}
    ]
      \begin{scope}
        \node[node] (a) at (0,0.0) {$a$};
        \node[node] (b) at (2,+1) {$b$};
        \node[node] (c) at (2,-1) {$c$};
        \node[node] (d) at (5,+1) {$d$};
        \node[node] (e) at (5,-1) {$e$};
        \draw[arc] (a) to node[above left] {$w_1 w_2 + w_2^2$} (b);
        \draw[arc] (a) to node[below left] {$w_1 w_2$} (c);
        \draw[arc] (b) to node[above] {$0$} (d);
        \draw[arc] (b) to node[below, near start] {$0$} (e);
        \draw[arc] (c) to node[below, near end] {$0$} (d);
        \draw[arc] (c) to node[below] {$w_1^2 + w_1 w_2 + w_2^2$} (e);
      \end{scope}

      \node[anchor=west] at (7,0) {
        \begin{tabular}{l|l|l}
          Player    & \colorOne{1}            & \colorTwo{2} \\ \hline
          Source    & \colorOne{$a$}          & \colorTwo{$a$} \\ \hline
          Sink      & \colorOne{$d$}          & \colorTwo{$e$} \\ \hline
          Possible  & \colorOne{$a$-$b$-$d$}, & \colorTwo{$a$-$b$-$e$},  \\
          actions   & \colorOne{$a$-$c$-$d$}  & \colorTwo{$a$-$c$-$e$}
        \end{tabular}
      };
    \end{tikzpicture}
  \end{center}
  \caption{%
    Simultaneous network routing game with weights $w_1, w_2$ for which the price of anarchy for uniform costs (as well as that for proportional costs; see~Section~\ref{sec_sim_prop}) is worst possible in case $w_1 \geq w_2$.
    Depicted are the network with cost coefficients (left) as well as the sources, sinks and possible actions of each player (right).
  }
  \label{fig_sim_uniprop_instance}
\end{figure}
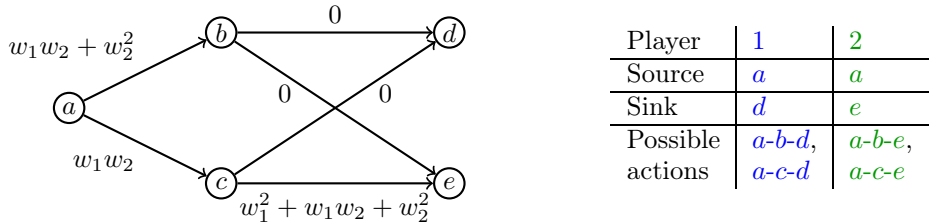

\subsection{Uniform costs}
\label{sec_sim_uni}

\begin{theorem}
  \label{thm_sim_uni_specific_weights}
  Let $w_1, w_2 \geq 0$ with $w_1 + w_2 > 0$.
  The price of anarchy for simultaneous uniformly weighted two-player congestion/network routing games with weights $w_1, w_2$ is equal to
  \[
    1 + \frac{ 2w_1w_2 + \max(w_1,w_2)^2 }{ w_1^2 + w_1w_2 + w_2^2 }.
  \]
  It is attained by the weighted network routing game described in \cref{fig_sim_uniprop_instance}.
\end{theorem}

The values of the expression in \cref{thm_sim_uni_specific_weights} are plotted in \cref{fig_plot}.
By taking the supremum over all feasible weight pairs, we also obtain the following result for the case in which the weights are not fixed.

\begin{corollary}
  \label{thm_sim_uni_all_weights}
  The price of anarchy for simultaneous uniformly weighted two-player congestion/network routing games is equal to $1 + 2/\sqrt{3} \approx 2.155$, which is attained for $w_1 / w_2 = 1 + \sqrt{3} \approx 2.732$ and for $w_2 / w_1 = 1 + \sqrt{3} \approx 2.732 $.
\end{corollary}

\begin{proof}[Proof of \cref{thm_sim_uni_all_weights}]
  By symmetry with respect to weights we can assume $w_1 \geq w_2$.
  Moreover, scaling $w_1$ and $w_2$ by the same positive scalar does not change the price of anarchy and hence we can assume $w_2 = 1$.
  Thus, the maximum is equal to that of the function $f(x) = 1 + \frac{ x^2+2x }{ x^2+x+1 } = 2 + \frac{ x - 1 }{ x^2+x+1 }$ over $x \in [1,\infty)$.
  It is easily verified that $f$ attains its global maximum at $x = 1 + \sqrt{3}$.
  \qed
\end{proof}

\begin{proof}[Proof of the lower bound in \cref{thm_sim_uni_specific_weights}]
  Consider the network routing game described in \cref{fig_sim_uniprop_instance}.
  The case in which the players choose the paths $a$-$c$-$d$ and $a$-$b$-$e$, respectively, has cost $w_1^2 w_2 + w_1 w_2^2 + w_2^3$.
  Hence, this value is an upper bound on the cost of the social optimum.

  We claim that the actions $a$-$b$-$d$ and $a$-$c$-$e$ (for players $1$ and $2$, respectively) constitute a Nash equilibrium.
  Player~1 has cost $w_1^2 w_2 + w_1 w_2^2$ and changing to $a$-$c$-$d$ would incur the same cost of $w_1 w_2 (w_1 + w_2)$ since arc $a$-$c$ would be used by both players.
  Player~2 has cost $w_1^2 w_2 + 2w_1 w_2^2 + w_2^3$ and changing to $a$-$b$-$e$ would incur the same cost of $w_1 w_2 + w_2^2 (w_1 + w_2)$ since arc $a$-$b$ would be used by both players.
  The total cost of the equilibrium is $2w_1^2 w_2 + 3w_1 w_2^2 + w_2^3$.
  This yields for the price of anarchy a lower bound of $\frac{(2w_1^2 + 3w_1 w_2 + w_2^2) \cdot w_2}{(w_1^2 + w_1 w_2 + w_2^2) \cdot w_2} = 1 + \frac{ w_1^2 + 2w_1 w_2 }{ w_1^2 + w_1 w_2 + w_2^2 }$, which concludes the proof.
  \qed
\end{proof}

\subsection{Proportional costs}
\label{sec_sim_prop}

\begin{theorem}
  \label{thm_sim_prop_specific_weights}
  Let $w_1, w_2 \geq 0$ with $w_1 + w_2 > 0$.
  The price of anarchy for simultaneous proportionally weighted two-player congestion/network routing games with weights $w_1, w_2$ is equal to
  \[
    1 + w_1w_2 \frac{ w_1 + w_2 + \max(w_1,w_2) }{ w_1^3 + w_2^3 + w_1w_2 \min(w_1, w_2) }.
  \]
  It is attained by the weighted network routing game described in \cref{fig_sim_uniprop_instance}.
\end{theorem}

We again obtain the result for arbitrary weights.

\begin{corollary}
  \label{thm_sim_prop_all_weights}
  The price of anarchy for simultaneous proportionally weighted two-player congestion/network routing games is approximately $2.0411$, which is attained for $w_1 / w_2 \approx 1.2704$ and for $w_2 / w_1 \approx 1.2704$.
\end{corollary}

\begin{proof}[Proof of \cref{thm_sim_prop_all_weights}]
  By symmetry with respect to weights we can assume $w_1 \geq w_2$.
  Moreover, scaling $w_1$ and $w_2$ by the same positive scalar does not change the price of anarchy and hence we can assume $w_2 = 1$.
  Thus, the maximum is equal to that of the function $f(x) = 1 + x \frac{ 2x + 1 }{ x^3 + x + 1 }$ whose derivative is
  $f'(x) = \frac{ -2x^4 - 2x^3 + 2x^2 + 4x + 1 }{ (x^3 + x+ 1)^2}$.
  The exact expression for the only relevant root of $f'$ is complicated, and hence we provide only the numerical solution $x^\star \approx 1.2704$ which yields an approximate price of anarchy of $2.0411$.
  \qed
\end{proof}

\begin{proof}[Proof of the lower bound in \cref{thm_sim_prop_specific_weights}]
  Consider the network routing game described in \cref{fig_sim_uniprop_instance}.
  The case in which the players choose the paths $a$-$c$-$d$ and $a$-$b$-$e$, respectively, has cost $w_1^3 w_2 + w_1 w_2^3 + w_2^4$.
  Hence, this value is an upper bound on the cost of the social optimum.

  We claim that the actions $a$-$b$-$d$ and $a$-$c$-$e$ (for players $1$ and $2$, respectively) constitute a Nash equilibrium.
  Player~1 has cost $w_1^3 w_2 + w_1^2 w_2^2$ and changing to $a$-$c$-$d$ would incur the same cost of $w_1^2 w_2 (w_1 + w_2)$ since arc $a$-$c$ would be used by both players.
  Player~2 has cost $w_1^2w_2^2 + 2w_1w_2^3 + w_2^4$ and changing to $a$-$b$-$e$ would incur the same cost of $(w_1 w_2 + w_2^2) (w_1 + w_2) w_2$ since arc $a$-$b$ would be used by both players.
  The total cost of the equilibrium is $w_1^3 w_2 + 2w_1^2 w_2^2 + 2w_1w_2^3 + w_2^4$.
  This yields for the price of anarchy a lower bound of $\frac{ (w_1^3 + 2w_1^2 w_2 + 2w_1w_2^2 + w_2^3) w_2 }{ (w_1^3 + w_1 w_2^2 + w_2^3) w_2 } = 1 + w_1 w_2 \frac{ w_1 + 2w_2 }{ w_1^3 + w_2^3 + w_1 w_2^2  }$, which concludes the proof.
  \qed
\end{proof}

\section{Results for symmetric simultaneous games}
\label{sec_sym}

\subsection{Uniform costs}
\label{sec_sym_uni}


For uniform cost functions we have an irrational weight ratio at which the behavior changes.
We denote it by $\tau \coloneqq \frac{1}{3} (2  \sqrt[3]{ 62 - 3\sqrt{183} } + \sqrt[3]{ 62 + 3\sqrt{183} } \approx 3.1527$.

\begin{theorem}
  \label{thm_sym_uni_specific_weights}
  Let $w_1, w_2 \geq 0$ with $w_1 + w_2 > 0$.
  The price of anarchy for symmetric simultaneous uniformly weighted two-player congestion games with weights $w_1, w_2$ is equal to
  \begin{enumerate}[label=(\alph*)]
  \item
    \label{sym_uni_weights1}
    $\frac{3w_1^3 + 9w_1^2w_2 + 9 w_1w_2^2 + 3w_2^2}{ 2w_1^3 + 5w_1^2w_2 + 5w_1w_2^2 + 2w_2^3 + w_1 w_2 \min(w_1,w_2) }$ if $\frac{1}{2}w_2 \leq w_1 \leq 2w_2$,
  \item
    \label{sym_uni_weights2}
    $\frac{2w_1^3 + 5w_1^2w_2 + 5w_1w_2^2 + 2w_2^3 + 3\max(w_1^2w_2 + w_1^3,w_1w_2^2 + w_2^3) }{2w_1^3 + 4w_1^2w_2 + 4w_1w_2^2 + 2w_2^3}$ if $2w_2 \leq w_1 \leq \tau w_2$ or if $\frac{1}{\tau} w_2 \leq w_1 \leq \frac{1}{2}w_2$.
  \item
    \label{sym_uni_weights3}
    $\frac{2w_1^2 + 2w_1w_2 + 2w_2^2 }{ (w_1+w_2)^2 }$ if $w_1 \geq \tau w_2$ or if $w_1 \leq \frac{1}{\tau} w_2$.
  \end{enumerate}
\end{theorem}

\begin{corollary}
  \label{thm_sym_uni_all_weights}
  The price of anarchy for symmetric simultaneous uniformly weighted two-player congestion games is equal to $2$.
\end{corollary}


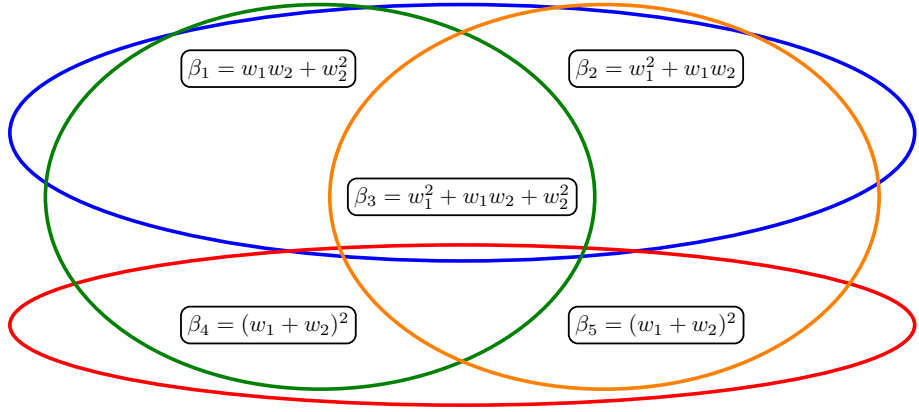
\begin{figure}[htpb]
  \begin{center}
    \scalebox{0.85}{%
    \begin{tikzpicture}[
      resource/.style={draw,thick,rounded corners,rectangle,inner sep=1mm,minimum size=2mm},
    ]
      \node[resource] (r1) at (-3,2) {$\beta_1 = w_1w_2 + w_2^2$};
      \node[resource] (r2) at (+3,2) {$\beta_2 = w_1^2 + w_1w_2$};
      \node[resource] (r3) at (0,0) {$\beta_3 = w_1^2 + w_1w_2 + w_2^2$};
      \node[resource] (r4) at (-3,-2) {$\beta_4 = (w_1+w_2)^2$};
      \node[resource] (r5) at (+3,-2) {$\beta_5 = (w_1+w_2)^2$};

      \node[ultra thick, draw=blue, ellipse, minimum width=140mm, minimum height=40mm] at (0.0,1.0) {}; 
      \node[ultra thick, draw=red, ellipse, minimum width=140mm, minimum height=25mm] at (0.0,-2.0) {}; 
      \node[ultra thick, draw=green!50!black, ellipse, minimum width=85mm, minimum height=60mm] at (-2.2,0) {}; 
      \node[ultra thick, draw=orange, ellipse, minimum width=85mm, minimum height=60mm] at (+2.2,0.0) {}; 
    \end{tikzpicture}%
    }
  \end{center}
  \caption{%
    Symmetric simultaneous uniformly weighted congestion game with weights $w_1, w_2$ for which the price of anarchy is worst possible in case $w_2 \leq w_1 \leq 2w_2$ holds.
    Depicted are the four possible actions indicating which of the five resources $1,2,\dotsc,5$ are chosen.
    Coefficients $\alpha_r$ have value $0$.
  }
  \label{fig_sym_uni_instance1}
\end{figure}

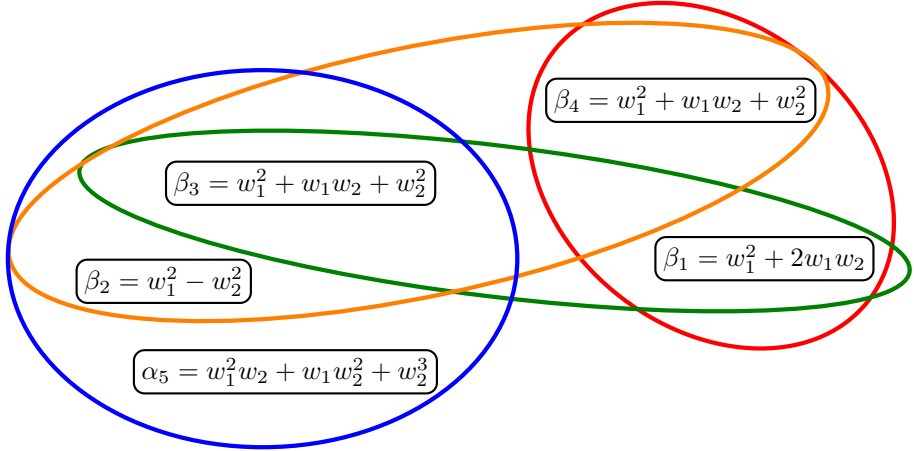
\begin{figure}[htpb]
  \begin{center}
    \begin{tikzpicture}[
      resource/.style={draw,rounded corners,thick,rectangle,inner sep=1mm,minimum size=2mm},
    ]
      \node[resource] (r1) at (3.1,0) {$\beta_1 = w_1^2 +2w_1w_2$};
      \node[resource] (r2) at (-4.8,-0.3) {$\beta_2 = w_1^2 - w_2^2$};
      \node[resource] (r3) at (-3,1.0) {$\beta_3 = w_1^2 + w_1w_2 +w_2^2$};
      \node[resource] (r4) at (2,2.1) {$\beta_4 = w_1^2 + w_1w_2 +w_2^2$};
      \node[resource] (r5) at (-3.2,-1.5) {$\alpha_5 = w_1^2w_2 + w_1w_2^2 + w_2^3$};

      \node[ultra thick, draw=red, ellipse, minimum width=40mm, minimum height=53mm, rotate=50] at (2.4,1.1) {};
      \node[ultra thick, draw=green!50!black, ellipse, minimum width=110mm, minimum height=20mm, rotate=-7] at (-0.45,0.5) {};
      \node[ultra thick, draw=orange, ellipse, minimum width=110mm, minimum height=33mm, rotate=12] at (-1.45,1.15) {};
      \node[ultra thick, draw=blue, ellipse, minimum width=67mm, minimum height=50mm,rotate=0] at (-3.5,-0) {};
    \end{tikzpicture}
  \end{center}
  \caption{%
    Symmetric simultaneous uniformly weighted congestion game with weights $w_1, w_2$ for which the price of anarchy is worst possible in case $2w_2 \leq w_1 \leq \tau w_2$ holds.
    Depicted are the four possible actions indicating which of the five resources $1,2,\dotsc,5$ are chosen.
    Coefficients $\alpha_r$ and $\beta_r$ that are not shown have value $0$.
  }
  \label{fig_sym_uni_instance2}
\end{figure}

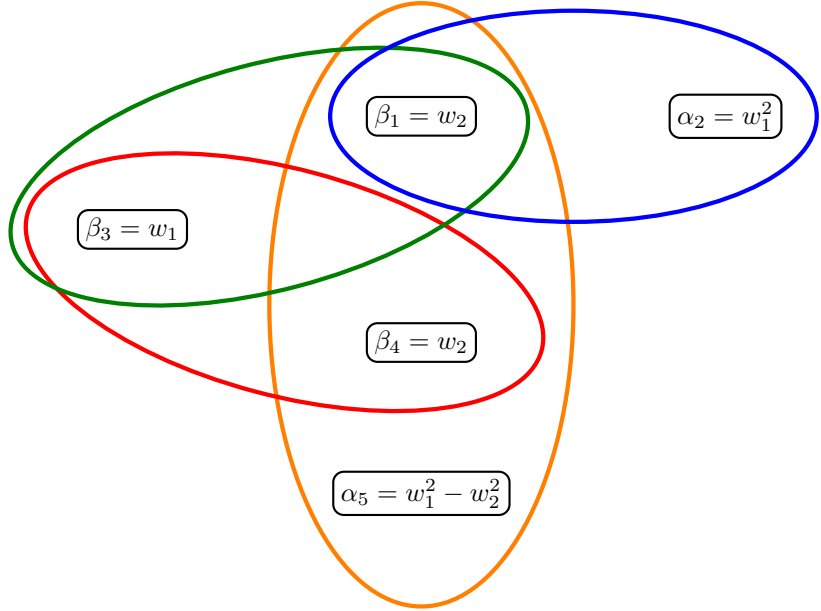
\begin{figure}[htpb]
  \begin{center}
    \begin{tikzpicture}[
      resource/.style={draw,thick,rounded corners,rectangle,inner sep=1mm,minimum size=2mm},
    ]
      \node[resource] (r1) at (0,3) {$\beta_1 = w_2$};
      \node[resource] (r2) at (4,3) {$\alpha_2 = w_1^2$};
      \node[resource] (r3) at (-3.8,1.5) {$\beta_3 = w_1$};
      \node[resource] (r4) at (0,0) {$\beta_4 = w_2$};
      \node[resource] (r5) at (0,-2) {$\alpha_5 = w_1^2 - w_2^2$};

      \node[ultra thick, draw=orange, ellipse, minimum width=40mm, minimum height=80mm] at (0,0.5) {};
      \node[ultra thick, draw=red, ellipse, minimum width=70mm, minimum height=30mm, rotate=-15] at (-1.8,0.8) {};
      \node[ultra thick, draw=green!50!black, ellipse, minimum width=70mm, minimum height=30mm, rotate=15] at (-2,2.2) {};
      \node[ultra thick, draw=blue, ellipse, minimum width=64mm, minimum height=28mm] at (2.0,3.0) {};
    \end{tikzpicture}
  \end{center}
  \caption{%
    Symmetric simultaneous uniformly weighted congestion game with weights $w_1, w_2$ for which the price of anarchy is worst possible in case $w_1 \geq \tau w_2$ holds.
    Depicted are the four possible actions indicating which of the five resources $1,2,\dotsc,5$ are chosen.
    Coefficients $\alpha_r$ and $\beta_r$ that are not shown have value $0$.
  }
  \label{fig_sym_uni_instance3}
\end{figure}

\begin{proof}[Proof of lower bound in \cref{thm_sym_uni_specific_weights}]
  By symmetry we only need to prove the bounds for the case $w_1 \geq w_2$.
  We distinguish the remaining intervals for $w_1/w_2$:
  
\medskip
\noindent
  \textbf{Case $w_1 \leq 2w_2$:}
  Consider the congestion game described in \cref{fig_sym_uni_instance1}.
  The case in which the players choose the actions $\{1,2,3\}$ and $\{4,5\}$, respectively, has cost
  $ 2w_1^3 + 5w_1^2w_2 + 6w_1w_2^2 + 2w_2^2 $.
  Hence, this value is an upper bound on the cost of the social optimum.

  We claim that the actions $\{1,3,4\}$ and $\{2,3,5\}$ (for players $1$ and $2$, respectively) constitute a Nash equilibrium.
  Player~1 has cost
  $2w_1^3 + 5w_1^2w_2 + 4w_1w_2^2 + w_2^3 $,
  changing to $\{1,2,3\}$ or to $\{4,5\}$ incurs the same cost, and changing to $\{2,3,5\}$ incurs the cost
  $3w_1^3 + 7w_1^2w_2 + 6w_1w_2^2 + w_2^3$.
  Note that the latter cost is not smaller.
  Player~2 has cost
  $w_1^3 + 4w_1^2w_2 + 5w_1w_2^2 + 2w_2^3$, changing to $\{4,5\}$ incurs the same cost, changing to $\{1,2,3\}$ incurs the cost $w_1^3 + 5w_1^2w_2 + 4w_1w_2^2 + 2w_2^3$, and changing to $\{1,3,4\}$ incurs the cost $2w_1^3 + 7w_1^2w_2 + 6w_1w_2^2 + 3w_2^3$.
  Note that the latter costs are not smaller since $w_1 \geq w_2$ holds.
  The total cost of the equilibrium is $3w_1^3 + 9w_1^2w_2 + 8 w_1w_2^2 + 3w_2^2$.
  This yields for the price of anarchy a lower bound of $\frac{3w_1^3 + 9w_1^2w_2 + 9 w_1w_2^2 + 3w_2^2}{ 2w_1^3 + 5w_1^2w_2 + 6w_1w_2^2 + 2w_2^2 }$.

\medskip
\noindent
  \textbf{Case $2w_2 \leq w_1 \leq \tau$:}
  Consider the congestion game described in \cref{fig_sym_uni_instance2}.
  The case in which the players choose the actions $\{2,3,5\}$ and $\{1,4\}$, respectively, has cost $2w_1^3 + 4w_1^2w_2 + 4w_1w_2^2 + 2w_2^3$.
  Hence, this value is an upper bound on the cost of the social optimum.


  We claim that the actions $\{1,3\}$ and $\{2,3,4\}$ (for players $1$ and $2$, respectively) constitute a Nash equilibrium.
  Player~1 has cost $2w_1^3+4w_1^2w_2+2w_1w_2^2+w_2^3$, changing to $\{2,3,5\}$ or to $\{1,4\}$ incurs the same cost, and changing to $\{2,3,4\}$ incurs the cost $3w_1^3 + 5w_1^2w_2 + 3w_1w_2^2 + 2w_2^3$ which is clearly not smaller.
  Player~2 has cost $w_1^3+4w_1^2w_2+3w_1w_2^2+w_2^3$, changing to $\{2,3,5\}$ incurs the same cost, changing to $\{1,3\}$ incurs the cost $3w_1^3+4w_1^2w_2+4w_1w_2^2+w_2^3$, and changing to $\{1,4\}$ incurs the cost $2w_1^3+3w_1^2w_2+3w_1w_2^2+1w_2^3$.
  The last two cost terms are easily seen to not be smaller since $w_1 \geq w_2$ holds.
  The total cost of the equilibrium is $3w_1^3+8w_1^2w_2+5w_1w_2^2+w_2^3$.
  This yields for the price of anarchy a lower bound of $\frac{3w_1^3 + 8w_1^2w_2 + 5w_1w_2^2 + 2w_2^3}{2w_1^3 + 4w_1^2w_2 + 4w_1w_2^2 + 2w_2^3}$.
  
\medskip
\noindent
  \textbf{Case $w_1 \geq \tau w_2:$}
  Consider the congestion game described in \cref{fig_sym_uni_instance3}.
  The case in which the players choose the actions $\{1,2\}$ and $\{3,4\}$, respectively, has cost $w_1w_2 + w_1^2 + w_1 w_2 + w_2^2$.
  Hence, this value is an upper bound on the cost of the social optimum.

  We claim that the actions $\{1,3\}$ and $\{1,4,5\}$ (for players $1$ and $2$, respectively) constitute a Nash equilibrium.
  Player~1 has cost $w_1^2 + w_1 w_2 + w_2^2$, changing to $\{1,2\}$ incurs the same cost, changing to $\{3,4\}$ incurs the cost $2w_1^2 + w_1 w_2$, and changing to $\{1,4,5\}$ incurs the cost $w_1^2 - w_2^2 + (w_1 + w_2)^2 = 2w_1^2 + 2w_1w_2$.
  Note that the latter costs are not smaller since $w_1 \geq w_2$ holds.
  Player~2 has cost $w_1^2 + w_1w_2 + w_2^2$, changing to $\{1,2\}$ or to $\{3,4\}$ incurs the same cost, and changing to $\{1,3\}$ incurs the cost $w_1^2 + 2w_1w_2 + w_2^2$.
  The total cost of the equilibrium is $2w_1^2 w_2 + 2w_1 w_2 + 2w_2^2$.
  This yields for the price of anarchy a lower bound of $\frac{2w_1^2 + 2w_1w_2 + 2w_2^2}{w_1^2 + 2w_1w_2 + w_2^2}$.
  \qed
\end{proof}

Our computed worst case instances for which the price of anarchy approaches its supremum are no network routing games.
However, there is a family of such games for which the price of anarchy also approaches $2$, although slower than that for congestion games.

\begin{theorem}
  \label{thm_sym_uni_network}
  The price of anarchy for symmetric simultaneous uniformly weighted two-player network routing games is equal to $2$.
\end{theorem}

\begin{figure}[htb]
  \begin{center}
    \begin{tikzpicture}[
      node/.style={draw,thick,circle,inner sep=0mm,minimum size=4mm},
      arc/.style={draw,->,thick}
    ]
      \begin{scope}
        \node[node] (s) at (0,0) {$s$};
        \node[node] (a) at (2,1) {$a$};
        \node[node] (t) at (4,0) {$t$};
        \draw[arc] (s) to node[above left] {$0 ; w_1$} (a);
        \draw[arc] (a) to node[above right] {$0 ; w_2$} (t);
        \draw[arc] (s) to node[below] {$1 ; 0$} (t);
      \end{scope}

      \node[anchor=west] at (6,0.5) {
        \begin{tabular}{l|l|l}
          Player    & \colorOne{1}            & \colorTwo{2} \\ \hline
          Source    & \colorOne{$s$}          & \colorTwo{$s$} \\ \hline
          Sink      & \colorOne{$t$}          & \colorTwo{$t$} \\ \hline
          Possible  & \colorOne{$s$-$a$-$t$}, & \colorTwo{$s$-$a$-$t$},  \\
          actions   & \colorOne{$s$-$t$}      & \colorTwo{$s$-$t$}
        \end{tabular}
      };
    \end{tikzpicture}
  \end{center}
  \caption{%
    Symmetric simultaneous uniformly weighted network routing game with weights $w_1, w_2$ for which the price of anarchy approaches $2$ when $w_1 / w_2 \to \infty$.
    Depicted are the network with cost coefficients (left) as well as the source, sink and possible actions of both players (right).
    For each arc $r$ we denote by $\beta_r ; \alpha_r$ the linear and absolute terms of the uniform cost function.
  }
  \label{fig_sym_uni_network_instance}
\end{figure}

\begin{proof}
  The upper bound of $2$ follows from \cref{thm_sym_uni_all_weights} since every network routing game is also a congestion game.
  For the lower bound we consider the network routing game described in \cref{fig_sym_uni_network_instance}.
  The case in which the players choose the paths $s$-$a$-$t$ and $s$-$t$, respectively, has cost $w_1 + 2w_2$.
  Hence, this value is an upper bound on the cost of the social optimum.

  The situation in which both players use the arc $s$-$t$ constitutes a Nash equilibrium since each player has cost $w_1 + w_2$ and changing to the path $s$-$a$-$t$ would incur the same cost.
  The total cost of the equilibrium is $2w_1 + 2w_2$.
  This yields for the price of anarchy a lower bound of $\frac{2w_1 + 2w_2}{ w_1 + 2w_2 }$.
  This fraction approaches $2$ for $w_1/w_2 \to \infty$, which concludes the proof.
\end{proof}

\subsection{Proportional costs}
\label{sec_sym_prop}

Also for proportional cost functions we have an irrational weight ratio at which the behavior changes.
It is the unique real root $\sigma$ of the polynomial $x^4 - 3x^2 - 3x - 1$ and yields $\sigma \approx 2.14790$.

\begin{theorem}
  \label{thm_sym_prop_specific_weights}
  Let $w_1, w_2 \geq 0$ with $w_1 + w_2 > 0$.
  The price of anarchy for symmetric simultaneous proportionally weighted two-player congestion/network routing games with weights $w_1, w_2$ is equal to
  \begin{enumerate}[label=(\alph*)]
  \item
    \label{sym_prop_weights1}
    $\frac{ 2w_1^4 + 6w_1^3w_2 + 8w_1^2w_2^2 + 6w_1w_2^3 + 2w_2^4 }{ 2w_1^4 + 3w_1^3w_2 + 4w_1^2w_2^2 + 3w_1w_2^3 + 2w_2^4 + w_1w_2 \min(w_1^2,w_2^2) }$ if $\frac{w_2}{\sigma} \leq w_1 \leq \sigma w_2$.
  \item
    \label{sym_prop_weights2}
    $\frac{ 2w_1^3w_2 + 4w_1^2w_2^2 + 2w_1w_2^3 + 2\max(w_1^4+w_1^3w_2,w_1w_2^3+w_2^4) }{ w_1^3w_2 + 2w_1^2w_2^2 + 2w_1w_2^3 + w_2^4   + 2w_1w_2 \min(w_1^2,w_2^2) + \max(w_1^4,w_2^4) }$ if $w_1 \geq \sigma w_2$ or $w_1 \leq \frac{w_2}{\sigma}$.
  \end{enumerate}
\end{theorem}

We again obtain the result for arbitrary weights.

\begin{corollary}
  \label{thm_sym_prop_all_weights}
  The price of anarchy for symmetric simultaneous proportionally weighted two-player congestion/network routing games is approximately $1.6096$, which is attained for $w_1 / w_2 \approx 1.1940$ and for $w_2 / w_1 \approx 1.1940$.
\end{corollary}

\begin{proof}[Proof of \cref{thm_sym_prop_all_weights}]
  By symmetry with respect to weights we can assume $w_1 \geq w_2$.
  Moreover, scaling $w_1$ and $w_2$ by the same positive scalar does not change the price of anarchy and hence we can assume $w_2 = 1$.
  Thus, the maximum is equal to that of the function $f(x) = \frac{ 2x^4+6x^3+8x^2+6x+2 }{ 2x^4+3x^3+4x^2+4x+2 }$ whose derivative has a numerator of degree $6$.
  Hence we provide only the numerical solution $x^\star \approx 1.1940$ which yields an approximate price of anarchy of $1.6096$.
  \qed
\end{proof}

\begin{figure}[H]
  \begin{center}
    \begin{tikzpicture}[
      node/.style={draw,thick,circle,inner sep=0mm,minimum size=4mm},
      arc/.style={draw,->,thick}
    ]
      \begin{scope}
        \node[node] (s) at (-0.1,-0.3) {$s$};
        \node[node] (a) at (2.7,-2.1) {$a$};
        \node[node] (b) at (2.1,0) {$b$};
        \node[node] (c) at (5.1,0) {$c$};
        \node[node] (d) at (4.5,-2.1) {$d$};
        \node[node] (t) at (7.9,-0.8) {$t$};
        \draw[arc] (s) to node[below, rotate=-33] {$w_1^2 + 2w_1w_2 + w_2^2$} (a);
        \draw[arc] (s) to node[above, rotate=7] {$w_1w_2 + w_2^2$} (b);
        \draw[arc] (a) to node[right] {$0$} (b);
        \draw[arc] (a) to node[above] {$0$} (d);
        \draw[arc] (c) to node[left] {$0$} (d);
        \draw[arc] (b) to node[above] {$w_1^2+w_1w_2+w_2^2$} (c);
        \draw[arc] (d) to node[below, rotate=23] {$w_1^2 + 2w_1w_2 + w_2^2$} (t);
        \draw[arc] (c) to node[above, rotate=-15] {$w_1^2 + w_1w_2$} (t);
      \end{scope}

      \node[anchor=north east] at (-0.6,0.7) {
        \begin{tabular}{l|l}
          Player    & \colorOne{1} and \colorTwo{2} \\ \hline
          Source    & $s$ \\ \hline
          Sink      & $t$ \\ \hline
          Possible  & $s$-$b$-$c$-$t$,  \\
          actions   & $s$-$b$-$c$-$d$-$t$, \\
                    & $s$-$a$-$b$-$c$-$t$, \\
                    & $s$-$a$-$d$-$t$, \\
                    & $s$-$a$-$b$-$c$-$d$-$t$
        \end{tabular}
      };
    \end{tikzpicture}
  \end{center}
  \caption{%
    Symmetric simultaneous network routing game with weights $w_1, w_2$ for which the price of anarchy for proportional weights is worst possible in case $w_2 \leq w_1 \leq \sigma w_2$.
    Depicted are the network with cost coefficients (left) as well as the sources, sinks and possible actions of each player (right).
  }
  \label{fig_sym_prop_instance1}
\end{figure}
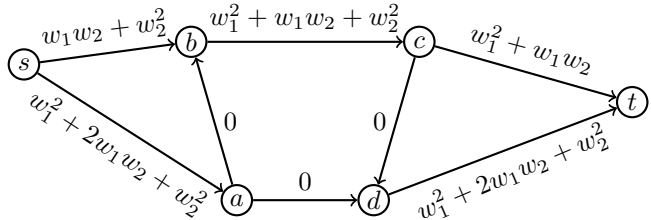

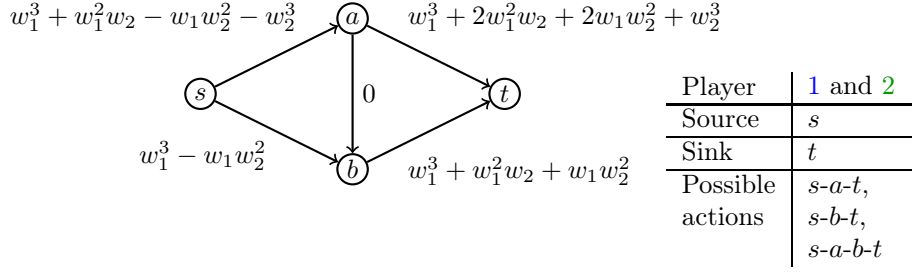
\begin{figure}[H]
  \begin{center}
    \begin{tikzpicture}[
      node/.style={draw,thick,circle,inner sep=0mm,minimum size=4mm},
      arc/.style={draw,->,thick}
    ]
      \begin{scope}
        \node[node] (s) at (0,0) {$s$};
        \node[node] (a) at (2,+1) {$a$};
        \node[node] (b) at (2,-1) {$b$};
        \node[node] (t) at (4,0) {$t$};
        \draw[arc] (s) to node[above left,near end] {$w_1^3 + w_1^2w_2 - w_1w_2^2 - w_2^3$} (a); 
        \draw[arc] (s) to node[below left] {$w_1^3 - w_1w_2^2$} (b);
        \draw[arc] (a) to node[right] {$0$} (b);
        \draw[arc] (a) to node[above right, near start] {$w_1^3+2w_1^2w_2+2w_1w_2^2+w_2^3$} (t);
        \draw[arc] (b) to node[below right, near start] {$w_1^3+w_1^2w_2+w_1w_2^2$} (t);
      \end{scope}

      \node[anchor=west] at (6,-1) {
        \begin{tabular}{l|l}
          Player    & \colorOne{1} and \colorTwo{2} \\ \hline
          Source    & $s$ \\ \hline
          Sink      & $t$ \\ \hline
          Possible  & $s$-$a$-$t$, \\
          actions   & $s$-$b$-$t$, \\
                    & $s$-$a$-$b$-$t$
        \end{tabular}
      };
    \end{tikzpicture}
  \end{center}
  \caption{%
    Symmetric simultaneous network routing game with weights $w_1, w_2$ for which the price of anarchy for proportional weights is worst possible in case $w_1 \geq \sigma w_2$.
    Depicted are the network with cost coefficients (left) as well as the sources, sinks and possible actions of each player (right).
  }
  \label{fig_sym_prop_instance2}
\end{figure}

\begin{proof}[Proof of lower bound in \cref{thm_sym_prop_specific_weights}]
  By symmetry we only need to prove the bounds for the case $w_1 \geq w_2$.
  We distinguish the remaining intervals for $w_1/w_2$:

\medskip
\noindent
  \textbf{Case $w_2 \leq w_1 \leq \sigma w_2$:}
  Consider the network routing game described in \cref{fig_sym_prop_instance1}.
  The case in which the players choose the path $s$-$b$-$c$-$t$ and $s$-$a$-$d$-$t$, respectively, has cost
  $2w_1^4 + 3w_1^3w_2 + 4w_1^2w_2^2 + 4w_1w_2^3 + 2w_2^4$.
  Hence, this value is an upper bound on the cost of the social optimum.

  We claim that the paths $s$-$b$-$c$-$d$-$t$ and $s$-$a$-$b$-$c$-$t$ (for players $1$ and $2$, respectively) constitute a Nash equilibrium.
  Player~1 has cost
  $2w_1^4 + 5w_1^3w_2 + 4w_1^2w_2^2 + w_1w_2^3$,
  changing to $s$-$b$-$c$-$t$ or to $s$-$a$-$d$-$t$
  incurs the same cost, and changing to $s$-$a$-$b$-$c$-$t$ or to $s$-$a$-$b$-$c$-$d$-$t$
  incurs the cost $3w_1^4 + 7w_1^3w_2 + 6w_1^2w_2^2 + 2w_1w_2^3$, which is not smaller due to $w_1,w_2 \geq 0$.
  Player~2 has cost
  $w_1^3w_2 + 4w_1^2w_2^2 + 5w_1w_2^3 + 2w_2^4$,
  changing to $s$-$b$-$c$-$t$ or to $s$-$a$-$d$-$t$
  incurs the same cost, and changing to $s$-$b$-$c$-$d$-$t$ or to $s$-$a$-$b$-$c$-$d$-$t$
  incurs the cost $2w_1^3w_2 + 6w_1^2w_2^2 + 7w_1w_2^3 + 3w_2^4$, which is not smaller due to $w_1,w_2 \geq 0$.
  The total cost of the equilibrium is 
  $2w_1^4 + 6w_1^3w_2 + 8w_1^2w_2^2 + 6w_1w_2^3 + 2w_2^4$.
  This yields for the price of anarchy a lower bound of $\frac{ 2w_1^4 + 6w_1^3w_2 + 8w_1^2w_2^2 + 6w_1w_2^3 + 2w_2^4 }{ 2w_1^4 + 3w_1^3w_2 + 4w_1^2w_2^2 + 4w_1w_2^3 + 2w_2^4 }$.
  
\medskip
\noindent
  \textbf{Case $w_1 \geq \sigma w_2$:}
  Consider the network routing game described in \cref{fig_sym_prop_instance2}.
  The case in which the players choose the path $s$-$b$-$t$ and $s$-$a$-$t$, respectively, has cost
  $2w_1^5 + w_1^4w_2 + 2w_1^3w_2^2 + 3w_1^2w_2^3 + w_1w_2^4$.
  Hence, this value is an upper bound on the cost of the social optimum.


  We claim that the paths $s$-$a$-$b$-$t$ and $s$-$b$-$t$ (for players $1$ and $2$, respectively) constitute a Nash equilibrium.
  Player~1 has cost
  $2w_1^5 + 3w_1^4w_2 + w_1^3w_2^2$,
  changing to $s$-$a$-$t$ or to $s$-$b$-$t$ both incur the
  same cost.
  Player~2 has cost
  $w_1^4w_2 + 3w_1^3w_2^2 + 2w_1^2w_2^3$, changing to $s$-$a$-$t$ incurs the
  same cost, and changing to $s$-$a$-$b$-$t$ incurs the cost
  $2w_1^4w_2 + 4w_1^3w_2^2 + 2w_1^2w_2^3-w_1w_2^4-w_2^5$.
  Note that the latter costs are not smaller since $w_1 \geq w_2$ holds.
  The total cost of the equilibrium is 
  $2w_1^5 + 4w_1^4w_2 + 4 w_1^3w_2^2 + 2w_1^2w_2^3$.
  This yields for the price of anarchy a lower bound of $\frac{ 2w_1^5 + 4w_1^4w_2 + 4 w_1^3w_2^2 + 2w_1^2w_2^3 }{ 2w_1^5 + w_1^4w_2 + 2w_1^3w_2^2 + 3w_1^2w_2^3 + w_1w_2^4 }$.
  \qed
\end{proof}

\section{Results for sequential games}
\label{sec_seq}

\subsection{Uniform costs}
\label{sec_seq_uni}

\begin{theorem}
  \label{thm_seq_uni_specific_weights}
  Let $w_1, w_2 \geq 0$ with $w_1 + w_2 > 0$.
  The price of anarchy for sequential uniformly weighted two-player congestion/network routing games with weights $w_1, w_2$ is equal to
  \begin{enumerate}[label=(\alph*)]
  \item
    \label{thm_seq_uni_weights1}
    $1 + \frac{w_1}{w_1 + w_2}$ if $w_2 \leq w_1$,
  \item
    \label{thm_seq_uni_weights2}
    $1 + \frac{2w_1w_2}{2w_1^2 + w_1w_2 + w_2^2}$ if $w_1 \leq w_2 \leq 2w_1$ and
  \item
    \label{thm_seq_uni_weights3}
    $1 + \frac{w_2}{2w_1+w_2}$ if $w_2 \geq 2w_1$.
  \end{enumerate}
  It is attained by the weighted network routing games described in \cref{fig_seq_uni_instance}.
\end{theorem}

By taking the supremum over all feasible weight pairs, we also obtain the following result for the case in which the weights are not fixed.

\begin{corollary}
  \label{thm_seq_uni_all_weights}
  The price of anarchy for sequential uniformly weighted two-player congestion games is equal to $2$.
\end{corollary}

\begin{proof}[Proof of \cref{thm_seq_uni_all_weights}]
  First, it is easy to see that~\ref{thm_seq_uni_weights1} is at most $2$ due to $w_2 \geq 0$.
  Second, \ref{thm_seq_uni_weights2} is at most $2$ since $0 \leq (w_1-w_2)^2 = w_1^2 - 2w_1w_2 + w_2^2$ together with $w_1^2 + w_1w_2 \geq 0$ implies $2w_1w_2 \leq 2w_1^2 + w_1w_2 + w_2^2$.
  Finally, \ref{thm_seq_uni_weights3} is at most $2$ due to $w_1 \geq 0$.
  Moreover, setting $w_2 = 1$ and letting $w_1 \rightarrow \infty$ yields a limit of $2$ in case~\ref{thm_seq_uni_weights1}.
  Similarly, $w_1 = 1$ and $w_2 \rightarrow \infty$ yields a limit of $2$ in case~\ref{thm_seq_uni_weights3}.
  Clearly, in the range $w_1 \leq w_2 \leq 2w_1$, the term in~\ref{thm_seq_uni_weights2} never gets close to $2$.
  \qed
\end{proof}

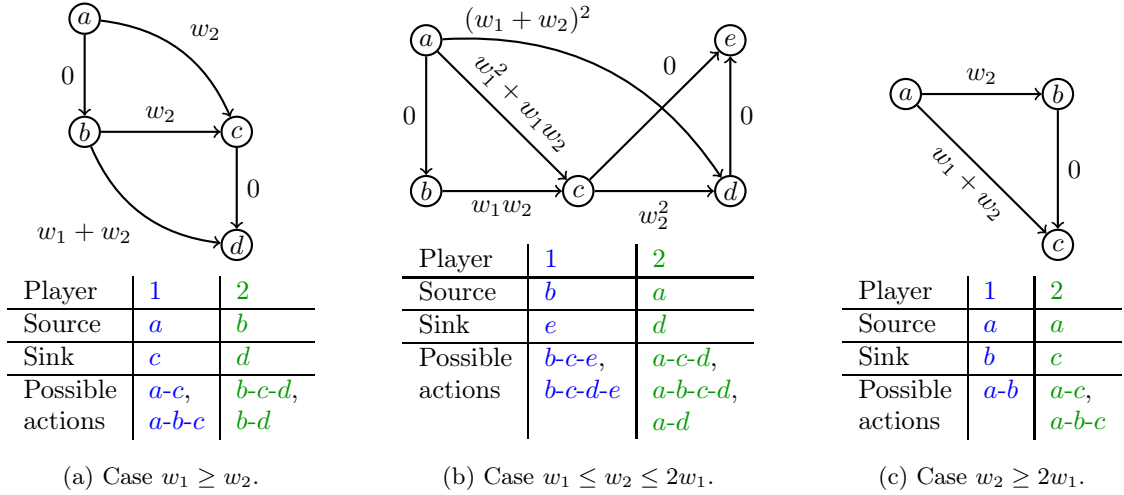
\begin{figure}[htpb]
  \subcaptionbox{Case~$w_1 \geq w_2$.\label{fig_seq_uni_instance1}}{%
    \begin{tikzpicture}[
      node/.style={draw,thick,circle,inner sep=0mm,minimum size=4mm},
      arc/.style={draw,->,thick}
    ]
      \begin{scope}
        \node[node] (s1) at (0,0.0) {$a$};
        \node[node] (s2) at (0,-1.5) {$b$};
        \node[node] (t1) at (2,-1.5) {$c$};
        \node[node] (t2) at (2,-3.0) {$d$};
        \draw[arc] (s1) to[bend left] node[above right] {$w_2$} (t1);
        \draw[arc] (s2) to[bend right] node[below left] {$w_1 + w_2$} (t2);
        \draw[arc] (s1) to node [left] {$0$} (s2);
        \draw[arc] (s2) to node [above] {$w_2$} (t1);
        \draw[arc] (t1) to node [right] {$0$} (t2);
      \end{scope}

      \node at (1,-4.5) {
        \begin{tabular}{l|l|l}
          Player    & \colorOne{1}           & \colorTwo{2} \\ \hline
          Source    & \colorOne{$a$}         & \colorTwo{$b$} \\ \hline
          Sink      & \colorOne{$c$}         & \colorTwo{$d$} \\ \hline
          Possible  & \colorOne{$a$-$c$},    & \colorTwo{$b$-$c$-$d$},  \\
          actions   & \colorOne{$a$-$b$-$c$} & \colorTwo{$b$-$d$}
        \end{tabular}
      };
    \end{tikzpicture}
  }
  \hfill
  \subcaptionbox{Case~$w_1 \leq w_2 \leq 2w_1$.\label{fig_seq_uni_instance2}}{%
    \begin{tikzpicture}[
      node/.style={draw,thick,circle,inner sep=0mm,minimum size=4mm},
      arc/.style={draw,->,thick}
    ]
      \begin{scope}
        \node[node] (a) at (0,2) {$a$};
        \node[node] (b) at (0,0) {$b$};
        \node[node] (c) at (2,0) {$c$};
        \node[node] (d) at (4,0) {$d$};
        \node[node] (e) at (4,2) {$e$};
        \draw[arc] (a) to node[left] {$0$} (b);
        \draw[arc] (b) to node[below] {$w_1w_2$} (c);
        \draw[arc] (a) to node[above,rotate=-45] {~$w_1^2 + w_1w_2$} (c);
        \draw[arc] (c) to node[below] {$w_2^2$} (d);
        \draw[arc] (d) to node[right] {$0$} (e);
        \draw[arc] (c) to node[above left, near end] {$0$} (e);
        \draw[arc] (a) to[bend left] node[above,near start] {$(w_1+w_2)^2$} (d);
      \end{scope}

      \node at (2,-2) {
        \begin{tabular}{l|l|l}
          Player    & \colorOne{1}                & \colorTwo{2} \\ \hline
          Source    & \colorOne{$b$}              & \colorTwo{$a$} \\ \hline
          Sink      & \colorOne{$e$}              & \colorTwo{$d$} \\ \hline
          Possible  & \colorOne{$b$-$c$-$e$},     & \colorTwo{$a$-$c$-$d$},  \\
          actions   & \colorOne{$b$-$c$-$d$-$e$}  & \colorTwo{$a$-$b$-$c$-$d$}, \\
                    &                             & \colorTwo{$a$-$d$}
        \end{tabular}
      };
    \end{tikzpicture}
  }
  \hfill
  \subcaptionbox{Case~$w_2 \geq 2w_1$.\label{fig_seq_uni_instance3}}{%
    \begin{tikzpicture}[
      node/.style={draw,thick,circle,inner sep=0mm,minimum size=4mm},
      arc/.style={draw,->,thick}
    ]
      \begin{scope}
        \node[node] (a) at (0,0.0) {$a$};
        \node[node] (b) at (2,0) {$b$};
        \node[node] (c) at (2,-2) {$c$};
        \draw[arc] (a) to node[above] {$w_2$} (b);
        \draw[arc] (b) to node[right] {$0$} (c);
        \draw[arc] (a) to node[below,rotate=-45] {$w_1+w_2$} (c);
      \end{scope}

      \node at (1,-3.5) {
        \begin{tabular}{l|l|l}
          Player    & \colorOne{1}        & \colorTwo{2} \\ \hline
          Source    & \colorOne{$a$}      & \colorTwo{$a$} \\ \hline
          Sink      & \colorOne{$b$}      & \colorTwo{$c$} \\ \hline
          Possible  & \colorOne{$a$-$b$}  & \colorTwo{$a$-$c$},  \\
          actions   &                     & \colorTwo{$a$-$b$-$c$}
        \end{tabular}
      };
    \end{tikzpicture}
  }
  \caption{%
    Sequential network routing games with weights $w_1,w_2$ for which the price of anarchy for uniform weights is worst possible.
    For each of the three cases we depict the network with cost coefficients (top) as well as the sources, sinks and possible actions of each player (bottom).
    The game in \cref{fig_seq_uni_instance1} is also worst possible for proportional weights and arbitrary $w_1,w_2$.
  }
  \label{fig_seq_uni_instance}
\end{figure}

\begin{proof}[Proof of the lower bound in \cref{thm_seq_uni_specific_weights}]
  Consider the network routing games described in \cref{fig_seq_uni_instance}.
  We consider the individual cases separately.
  
\medskip
\noindent
  \textbf{Case~\ref{thm_seq_uni_weights1}:}
  Considering the actions $a$-$c$ and $b$-$c$-$d$ for the two players in \cref{fig_seq_uni_instance1} reveals that the cost of the social optimum is at most $w_1 w_2 + w_2^2$.
  Inspection of the other costs of both players shows that $a$-$b$-$c$ is an optimal strategy for player~1 since we can assume that then player~2 optimally plays $b$-$d$.
  The total cost of this subgame perfect equilibrium is equal to $2w_1 w_2 + w_2^2$.
  This yields for the price of anarchy a lower bound of $\frac{2w_1 w_2 + w_2^2}{w_1 w_2 + w_2^2} = 1 + \frac{ w_1 \cdot w_2 }{(w_1 + w_2) \cdot w_2} = 1 + \frac{ w_1 }{w_1 + w_2}$.
  
\medskip
\noindent
  \textbf{Case~\ref{thm_seq_uni_weights2}:}
  Considering the actions $b$-$c$-$e$ and $a$-$c$-$d$ for the two players in \cref{fig_seq_uni_instance2} reveals that the cost of the social optimum is at most $2w_1^2 w_2 + w_1w_2^2 + w_2^3$.
  Inspection of the other costs of both players shows that $b$-$c$-$d$-$e$ is an optimal strategy for player~1 since we can assume that then player~2 optimally plays $a$-$d$.
  The total cost of this subgame perfect equilibrium is equal to $2w_1^2 w_2 + 3w_1 w_2^2 + w_2^3$.
  This yields for price of anarchy a lower bound of $\frac{2w_1^2 w_2 + 3w_1 w_2^2 + w_2^3}{2w_1^2 w_2 + w_1 w_2^2 + w_2^3} = 1 + \frac{2w_1 w_2^2}{2w_1^2 w_2 + w_1 w_2^2 + w_2^3} = 1 + \frac{ 2w_1 w_2 }{2w_1^2 + w_1 w_2 + w_2^2}$.

\medskip
\noindent
  \textbf{Case~\ref{thm_seq_uni_weights3}:}
  Considering the actions $a$-$b$ and $a$-$c$ for the two players in \cref{fig_seq_uni_instance3} reveals that the cost of the social optimum is at most $2w_1 w_2 + w_2^2$.
  Inspection of the other costs of both players shows that $a$-$b$ is an optimal strategy for player~1 and we can assume that then player~2 optimally plays $a$-$b$-$c$.
  The total cost of this subgame perfect equilibrium is equal to $2w_1 w_2 + 2w_2^2$.
  This yields for the price of anarchy a lower bound of $\frac{2w_1 w_2 + 2w_2^2}{2w_1 w_2 + w_2^2} = 1 + \frac{ w_2^2 }{2w_1 w_2 + w_2^2} = 1 + \frac{ w_2 }{2w_1 + w_2}$.
  \qed
\end{proof}

\subsection{Proportional costs}
\label{sec_seq_prop}

\begin{theorem}
  \label{thm_seq_prop_specific_weights}
  Let $w_1, w_2 \geq 0$ with $w_1 + w_2 > 0$.
  The price of anarchy for sequential proportionally weighted two-player congestion/network routing games with weights $w_1, w_2$ is equal to
  \[
    1 + \frac{w_1w_2}{w_1^2 + w_2^2}.
  \]
\end{theorem}

By taking the supremum over all feasible weight pairs, we also obtain the following result for the case in which the weights are not fixed.

\begin{corollary}
  \label{thm_seq_prop_all_weights}
  The price of anarchy for sequential proportionally weighted two-player congestion games is equal to $1.5$, which is attained if and only if $w_1 = w_2$ holds.
\end{corollary}

\begin{proof}[Proof of \cref{thm_seq_prop_all_weights}]
  From $0 \leq (w_1 - w_2)^2 = w_1^2 - 2w_1w_2  + w_2^2$ we obtain $2w_1w_2 \leq w_1^2 + w_2^2$ which readily implies $\frac{w_1w_2}{w_1^2 + w_2^2} \leq \frac{1}{2}$.
  Moreover, equality is attained if and only if $w_1 = w_2$ holds.
\end{proof}

\begin{proof}[Proof of the lower bound in \cref{thm_seq_prop_specific_weights}]
  Consider the network routing game described in \cref{fig_seq_uni_instance1}, but with a proportional cost function.
  The case in which the players choose the paths $a$-$c$ and $b$-$c$-$d$, respectively, has cost $w_1^2w_2 + w_2^3$.
  Hence, this value is an upper bound on the social optimum

  By inspecting the costs of both players for the other outcomes it becomes clear that $a$-$b$-$c$ is an optimal strategy for player~1 since we can assume that then player~2 optimally plays $b$-$d$.
  The total cost of this subgame perfect equilibrium is equal to $w_1^2w_2 + w_1w_2^2 + w_2^3$.
  This yields for the price of anarchy a lower bound of $\frac{(w_1^2 + w_1w_2 + w_2^2) \cdot w_2}{(w_1^2 + w_2^2) \cdot w_2} = 1 + \frac{w_1w_2}{w_1^2 + w_2^2}$, which concludes the proof.
\end{proof}

\section{LP based proofs}
\label{sec_lp}

\DeclareDocumentCommand\labels{}{\mathcal{A}}
\DeclareDocumentCommand\labelsOne{}{\mathcal{A}_1}
\DeclareDocumentCommand\labelsTwo{}{\mathcal{A}_2}

First, observe that for fixed weights $w_1$ and $w_2$ the cost functions~\eqref{eq_cost_uni} and~\eqref{eq_cost_prop} are linear.
The idea of computing the price of anarchy is greedy: we construct an LP for an instance of a game with a minimal set of actions $\actions$ that are required for a worst-case instance.
The LP has as variables the cost parameters $\alpha_r,\beta_r$ of the cost functions per resource $r \in R$.
This is finite, since as in~\cite{JongU19} we can argue that by pigeonhole principle at most $2^{|\actions|}$ resources are needed for any such an instance: if two resources appear in precisely the same actions, then these could be combined into one (adding their costs), which yields an instance with fewer resources but the same price of anarchy.
Given that the LP's cost functions per resource are yet undetermined, one can w.l.o.g.\ ``label'' actions as social optimum and equilibrium, respectively.
These labels are $\labels = \{ \labelOneOpt, \labelTwoOpt, \labelOneEqui, \labelTwoEqui \}$.
For sequential games we extend $\labels$ by the label $\labelTwoEquiX$ for an optimal action for player~2 after player~1 has played $\labelOneOpt$.
The observation above justifies to view resources as subsets of labels, i.e., $R = 2^{\labels}$.
Therefore, if $r \in A$, we can also write $A \in r$, as we view $r$ as the set of all actions $A \ni r$.
In addition to the cost variables $\alpha_r, \beta_r$ for each $r \in R$ the LP has variables for the costs $C_i(A_1,A_2)$ for $i=1,2$ and for all labels $A_j \in \labels_j$ that are admissible for player~$j \in \{1,2\}$.
We now introduce the constraints of the LP that ensure correctness of the labels and then prove that optimal solutions indeed correspond to worst-case instances.
We start with the basic (but incomplete) LP.
\begin{subequations}
  \label{eq_lp_common}
  \begin{alignat}{7}
    & \text{max } ~\mathrlap{ \costOne(\labelOneEqui, \labelTwoEqui) + \costTwo(\labelOneEqui, \labelTwoEqui) } \\
    & \text{s.t. }
      & \costOne(\labelOneOpt,\labelTwoOpt) + \costTwo(\labelOneOpt,\labelTwoOpt) &= 1 \label{eq_lp_common_normalization} \\
    & & \costOne(A_1,A_2) + \costTwo(A_1,A_2) &\geq 1 &\qquad&\forall A_1 \in \labelsOne,~ \forall A_2 \in \labelsTwo \label{eq_lp_common_opt} \\
    & & \alpha_r,\beta_r &\geq 0 &&\forall r \in R
  \end{alignat}
\end{subequations}
In case of uniform cost functions~\eqref{eq_cost_uni} we add
\begin{align}
  C_i(A_1,A_2) &= \sum_{\substack{r \in R \\ A_i \in r}} (\alpha_r + \beta_r \sum_{j : A_j \in r} w_j) && \forall A_1 \in \labelsOne,~ \forall A_2 \in \labelsTwo, ~i=1,2 \, , \label{eq_lp_uni}
\end{align}
while in case of proportional cost functions~\eqref{eq_cost_prop} we add
\begin{align}
  C_i(A_1,A_2) &= w_i \sum_{\substack{r \in R \\ A_i \in r}} (\alpha_r + \beta_r \sum_{j : A_j \in r} w_j) && \forall A_1 \in \labelsOne,~ \forall A_2 \in \labelsTwo, ~i=1,2 \, . \label{eq_lp_prop}
\end{align}

\paragraph{Simultaneous games.}
For simultaneous games we need to add the Nash inequalities~\eqref{eq_lp_nash} to enforce that $\labelOneEqui$ and $\labelTwoEqui$ form a Nash equilibrium.
For general games we define $\actionsOne \coloneqq \{ \labelOneOpt, \labelOneEqui \}$ and $\actionsTwo \coloneqq \{ \labelTwoOpt, \labelTwoEqui \}$, while for symmetric games we allow both players to use the union $\actionsOne \coloneqq \actionsTwo \coloneqq \{ \labelOneOpt,\labelOneEqui, \labelTwoOpt, \labelTwoEqui \}$ of these actions.

\paragraph{Sequential games.}
The following constraints model that $\labelTwoEqui$ and $\labelTwoEquiX$ are optimal actions for player~2 after player~1 has played $\labelOneEqui$ and $\labelOneOpt$, respectively, as well as the requirement that $\labelOneEqui$ is a subgame-perfect action for player~1.
\begin{subequations}
  \label{eq_lp_seq}
  \begin{align}
    \costTwo(\labelOneEqui,\labelTwoEqui) &\leq \costTwo(\labelOneEqui,A_2) &&\forall A_2 \in \actionsTwo \label{eq_lp_seq_equi_change_two} \\
    \costTwo(\labelOneOpt,\labelTwoEquiX) &\leq \costTwo(\labelOneOpt,A_2) &&\forall A_2 \in \actionsTwo \label{eq_lp_seq_opt_change_two} \\
    \costOne(\labelOneEqui,\labelTwoEqui) &\leq \costOne(\labelOneOpt,\labelTwoEquiX) \label{eq_lp_seq_change_one}
  \end{align}
\end{subequations}
We thus define $\actionsOne \coloneqq \{ \labelOneOpt, \labelOneEqui \}$ and $\actionsTwo \coloneqq \{ \labelTwoOpt, \labelTwoEqui, \labelTwoEquiX \}$.
Our approach does not work for symmetric sequential games.
The reason is that if the actions from $\labelsTwo$ are available to player~1 as well, then for each of these actions we would need to introduce new actions for the subgame-perfect action of player~2 in such a case, which in turn would be available to player~1, and so on.

We now establish correctness of all variants of this LP.
\begin{theorem}
  \label{thm_correctness}
  For fixed weights $w_1,w_2 \geq 0$ with $w_1 + w_2 > 0$ the LP~\eqref{eq_lp_common} for the action sets $\actionsOne,\actionsTwo$ specified above, and extended by either~\eqref{eq_lp_uni} or~\eqref{eq_lp_prop} and by either~\eqref{eq_lp_nash} or~\eqref{eq_lp_seq} computes the price of anarchy for affine (sequential or simultaneous or symmetric simultaneous) congestion games with these specific weights.
\end{theorem}

\begin{proof}
  Consider a class of games and an LP as defined in the theorem.

  We first argue that a primal solution actually represents a game and that its objective value corresponds to its price of anarchy.
  The resources of the game are given by those $r \in R$ for which (at least one of) $\alpha_r$ or $\beta_r$ is positive.
  The game has $|\actions|$ actions that are labeled according to $\labels$.
  Each action $A \in \actions$ uses exactly those resources $r \in R$ for which $A \in r$ holds.
  The cost function of each player is defined via~\eqref{eq_lp_uni} or~\eqref{eq_lp_prop}, and these constraints ensure that the cost variables $C_i(A_1,A_2)$ actually represent the players' costs for these actions $A_1$ and $A_2$.
  Constraints~\eqref{eq_lp_nash} or~\eqref{eq_lp_seq} ensure that $(\labelOneEqui,\labelTwoEqui)$ is actually an equilibrium.
  Moreover, constraints~\eqref{eq_lp_common_normalization} and~\eqref{eq_lp_common_opt} enforce that $\labelOneOpt$ and $\labelTwoOpt$ form a socially optimal strategy.
  Since the latter has cost $1$, the objective value $\costOne(\labelOneEqui, \labelTwoEqui) + \costTwo(\labelOneEqui, \labelTwoEqui)$ is the price-of-anarchy for this game.
  If we consider symmetric (simultaneous) games then $\actionsOne = \actionsTwo$ implies that the defined game is indeed symmetric.

  It remains to show that \emph{every} game can be represented by a primal solution.
  To this end, consider a game with resources $\bar{R}$, actions
  $\bar{\labelsOne},\bar{\labelsTwo} \subseteq 2^{\bar{R}}$ for the two players as well as cost coefficients $\bar{\alpha},\bar{\beta} \in \R^{\bar{R}}$.
  By scaling we can assume that the cost of the social optimum is equal to $1$.
  We now create a mapping $\pi$ from labels to the game's actions.
  To this end, let $\pi(\labelOneOpt) \in \bar{\labelsOne}$ and $\pi(\labelTwoOpt) \in \bar{\labelsTwo}$ be actions that constitute a social optimum.
  Moreover, we consider an equilibrium (either Nash or subgame-perfect) for which the price of anarchy of this game is attained and let $\pi(\labelOneEqui) \in \bar{\labelsOne}$ and $\pi(\labelTwoEqui) \in \bar{\labelsTwo}$ constitute such an equilibrium.
  In the sequential case, let $\pi(\labelTwoEquiX) \in \bar{\labelsTwo}$ be an action that is subgame-perfect for player~2 after player~1 has played $\pi(\labelOneOpt)$.

  The mapping $\pi$ of labels $\labels$ to actions $\pi(\actions)\subseteq \bar{\actions}$ identifies the ``relevant'' actions from $\bar{\labels}$. For any resource $\bar{r}\in\bar{R}$, we can now associate with it the ``incidence pattern'' of the set of actions from $\pi(\labels)$ in which $\bar{r}$ is contained.  This 
  induces a reverse mapping $\chi$ that maps each resource $\bar{r} \in \bar{R}$ to the unique $r \in R$  which has the same incidences.
  Formally, $A \in \chi(\bar{r}) \iff \bar{r} \in \pi(A)$ must hold for all $A \in \labels$.
  This allows us to aggregate the resources $\bar{R}$ accordingly via
  \begin{align}
    &&&&
    \alpha_r &\coloneqq \sum_{\bar{r}\ :\ r= \chi(\bar{r})} \bar{\alpha}_{\bar{r}}
    &\text{and}&&
    \beta_r &\coloneqq \sum_{\bar{r}\ :\ r= \chi(\bar{r})} \bar{\beta}_{\bar{r}} \, . \label{eq_mapping}
    &&&&
  \end{align}
  By equations~\eqref{eq_lp_uni} or~\eqref{eq_lp_prop}, the values of the remaining variables are determined uniquely.
  Moreover, \eqref{eq_mapping} ensures that for $i=1,2$ and for any profile $(A_1,A_2) \in \actionsOne \times \actionsTwo$, $C_i(A_1,A_2)$ is equal to the cost of player~$i$ if actions $\pi(A_1)$ and $\pi(A_2)$ are played in the game.
  In particular, constraints~\eqref{eq_lp_common_normalization}, \eqref{eq_lp_common_opt} as well as either~\eqref{eq_lp_nash} or~\eqref{eq_lp_seq} are satisfied.
  Consequently, the objective value corresponds to the price of anarchy of this game since the social optimum of the game equals $1$. Hence the optimal LP solution equals the price of anarchy for the given class of games.
\end{proof}

\paragraph{Arbitrary weights.}
In order to  derive the price of anarchy for arbitrary weights, the LPs from \cref{thm_correctness} can be used as an auxiliary tool.
First, an approximate version of \cref{fig_plot} can be produced, from which we could guess intervals of weight ratios $w_1/w_2$ for which the same LP basis is optimal.
Second, for each such interval, several weight pairs are chosen and optimal primal and dual solutions are computed.
Using the LP solver SoPlex, we computed exact rational solutions (see~\cite{GleixnerSW15}), which helped to derive educated guesses for algebraic expressions.

Before we provide a concrete example of such a derivation, we mention two tricks that were important.
First, cancellations in the expressions can be avoided by choosing prime numbers for the weights.
Second, we observed that often several optimal solutions exist, which makes it hard to make an educated guess of the weight-dependent expression for each variable.
We were able to circumvent this difficulty by forcing some of the cost variables to $0$ while ensuring that optimality of the solution is maintained.

\begin{table}[htpb]
  \caption{Nonzeros of optimal primal and dual solutions of LP~\eqref{eq_lp_common} extended with~\eqref{eq_lp_uni} and~\eqref{eq_lp_seq} for different weight combinations satisfying $2w_1 \leq w_2$.}
  \label{tab_educated_guess}
  \begin{center}
    \begin{tabular}{r|r|r|r|r|r|r|r|r}
      \multicolumn{2}{c|}{\textbf{Weights}} & \multicolumn{3}{c|}{\textbf{Primal optimal solution}} & \multicolumn{4}{c}{\textbf{Dual optimal solution}} \\
      $w_1$ & $w_2$ & $\beta_{\{ \labelOneEqui \}}$ & $\beta_{\{\labelTwoOpt, \labelTwoEqui \}}$ & $\beta_{\{ \labelOneOpt, \labelOneEqui, \labelTwoEquiX \}}$ & \eqref{eq_lp_common_normalization} &  \eqref{eq_lp_seq_equi_change_two}$_{\labelTwoEquiX}$ & \eqref{eq_lp_seq_opt_change_two}$_{\labelTwoOpt}$ & \eqref{eq_lp_seq_change_one}  \\ \hline
      2 & 14  & 7/18  & 4/63      & 1/18  & 16/9   & 1 & 16/9   & 1 \\
      1 & 7   & 7/9   & 8/63      & 1/9   & 16/9   & 1 & 16/9   & 1 \\
      1 & 13  & 13/15 & 14/195    & 1/15  & 28/15  & 1 & 28/15  & 1 \\
      1 & 100 & 50/51 & 101/10200 & 1/102 & 101/51 & 1 & 101/51 & 1 \\
    \end{tabular}
  \end{center}
\end{table}

We now demonstrate this method for the example of uniformly weighted sequential games.
By plotting the price of anarchy for many weight ratios $w_1/w_2$ (see \cref{fig_plot}), one can immediately conjecture that there are three different algebraic expressions, namely for the cases $w_1/w_2 \leq \frac{1}{2}$, $w_1/w_2 \in [\frac{1}{2},1]$ and for $w_1/w_2 \geq 1$.
We only show how we derived the expressions for weights satisfying $2w_1 \leq w_2$.
To this end, we compute exact LP solutions for different weight combinations.
To avoid cancellations in the expressions it is useful to choose prime numbers for the weights.
Moreover, there are often several optimal solutions which makes it hard to make an educated guess of the weight-dependent expression for each variable.
To avoid this, our implementation of the LP allows to force (resource) variables to $0$.
Of course, one has to ensure that this does not affect optimality.
Finally, our implementation does not actually have cost variables because they can be expressed as combinations of the $\alpha$ and $\beta$ variables.
\cref{tab_educated_guess} shows the primal and dual solutions for different weight combinations.
The first observation is that scaling $w_1$ and $w_2$ by $\lambda > 0$ implies a scaling of the primal solution by $1/\lambda$ but does not affect the dual solution.
This can easily be explained by the normalization of the social optimum to $1$.
A few guesses for the optimal solution values can be made.
We first read off
\begin{equation*}
  \beta_{\{\labelOneOpt,\labelOneEqui,\labelTwoEquiX}\} = \frac{1}{2w_1 + w_2} = \frac{w_1w_2}{2w_1^2w_2 + w_1w_2^2}.
\end{equation*}
Taking the denominator $2w_1 + w_2$ into account, the value
\begin{equation*}
\beta_{\{\labelOneEqui\}} = \frac{w_2/w_1}{2w_1 + w_2} = \frac{ w_2^2 }{ 2w_1^2w_2 + w_1w_2^2 }
\end{equation*}
is also easy to see.
The denominators of $\beta_{\{\labelTwoOpt,\labelTwoEqui\}}$ are all divisible by $2w_1 = w_2$ (considering $8/126$ instead of $4/63$ in the first row), and the factor is always equal to $w_2/w_1$.
From this we can conclude
\begin{equation*}
  \beta_{\{\labelTwoOpt,\labelTwoEqui\}} = \frac{ 1 + w_2/w_1 }{ (2w_1 + w_2) \cdot w_2/w_1 }
  = \frac{w_1^2 + w_1w_2}{2w_1^2w_2 + w_1w_2^2}
\end{equation*}
for this variable.
That the dual variables for~\eqref{eq_lp_seq_equi_change_two} (for $A_2 = \labelTwoEquiX$) and~\eqref{eq_lp_seq_change_one} are always $1$ is easy to see.
For those of~\eqref{eq_lp_common_normalization} and~\eqref{eq_lp_seq_opt_change_two} (for $A_2 = \labelTwoOpt$) we finally obtain $2w_1 + w_2$ as the denominator and $2w_1 + 2w_2$ for the numerator.

\paragraph{Algebraic proofs.}
Once we have conjectures for algebraic expressions for primal and dual solutions we can turn these into an algebraic proof.
This is a very mechanic procedure which is why we again describe it only for uniformly weighted sequential games with weights $2w_1 \leq w_2$.
In \cref{tab_example_costs} we state the costs that each player has to pay for every possible strategy.
It is easy to verify that $(\labelOneOpt,\labelTwoOpt)$ is a social optimum.
Moreover, we confirm that $\labelTwoEqui$ and $\labelTwoEquiX$ are subgame-perfect actions for player~2 after player~1 has chosen $\labelOneEqui$ or $\labelOneOpt$, respectively.
Consequently, $\labelOneEqui$ is a subgame-perfect action for player~1.
The resulting price of anarchy is $1 + \frac{2w_2}{ 2w_1 + w_2 }$ (cf.\ \cref{thm_seq_uni_specific_weights}).

\begin{table}[htbp]
  \caption{Costs for an optimal solution for uniformly weighted sequential games, scaled by the common denominator $\mu = 2w_1^2w_2 + w_1w_2^2$.}
  \label{tab_example_costs}
  \begin{center}
    \begin{tabular}{ll|ll|l}
      $A_1$ & $A_2$ & $\mu C_1(A_1,A_2)$ & $\mu C_2(A_1,A_2)$ & $\mu (C_1 + C_2)$ \\
      \hline
      $\labelOneOpt$ & $\labelTwoOpt$ & $w_1w_2 \cdot w_1$ & $(w_1^2+w_1w_2) \cdot w_2$ & $\mu$ \\
      $\labelOneOpt$ & $\labelTwoEqui$ & $w_1w_2 \cdot w_1$ & $(w_1^2+w_1w_2) \cdot w_2$ & $\mu$ \\
      $\labelOneOpt$ & $\labelTwoEquiX$ & $w_1w_2 \cdot (w_1+w_2)$ & $w_1w_2 \cdot (w_1+w_2)$ & $\mu + w_1w_2^2$ \\
      $\labelOneEqui$ & $\labelTwoOpt$ & $w_2^2 \cdot w_1 + w_1w_2 \cdot w_1$ & $(w_1^2+w_1w_2) \cdot w_2$ & $\mu + w_1w_2^2$ \\
      $\labelOneEqui$ & $\labelTwoEqui$ & $w_2^2 \cdot w_1 + w_1w_2 \cdot w_1$ & $(w_1^2+w_1w_2) \cdot w_2$ & $\mu + w_1w_2^2$ \\
      $\labelOneEqui$ & $\labelTwoEquiX$ & $w_2^2 \cdot w_1 + w_1w_2 \cdot (w_1+w_2)$ & $w_1w_2 \cdot (w_1 + w_2)$ & $\mu + 2w_1w_2^2$
    \end{tabular}
  \end{center}
\end{table}

The dual multipliers tell us how to combine the inequalities of the LP to a valid inequality that yields a matching upper bound on the optimum.
The corresponding combination of the inequalities reads
\begin{align*}
  \frac{2w_1+2w_2}{2w_1+w_2} \costOne(\labelOneOpt,\labelTwoOpt) + \frac{2w_1+2w_2}{2w_1+w_2} \costTwo(\labelOneOpt,\labelTwoOpt) &= \frac{2w_1+2w_2}{2w_1+w_2} \tag{\ref{eq_lp_common_normalization}} \\
  \costTwo(\labelOneEqui,\labelTwoEqui) &\leq \costTwo(\labelOneEqui,\labelTwoEquiX) \tag{\ref{eq_lp_seq_equi_change_two}} \\
  \frac{2w_1+2w_2}{2w_1+w_2} \costTwo(\labelOneOpt,\labelTwoEquiX) &\leq \frac{2w_1+2w_2}{2w_1+w_2} \costTwo(\labelOneOpt,\labelTwoOpt) \tag{\ref{eq_lp_seq_opt_change_two}} \\
  \costOne(\labelOneEqui,\labelTwoEqui) &\leq \costOne(\labelOneOpt,\labelTwoEquiX) \tag{\ref{eq_lp_seq_change_one}}
\end{align*}
Simplifying their sum yields
\begin{multline*}
  \frac{2w_1+2w_2}{2w_1+w_2} \costOne(\labelOneOpt,\labelTwoOpt)
  + \costTwo(\labelOneEqui,\labelTwoEqui)
  + \frac{2w_1+2w_2}{2w_1+w_2} \costTwo(\labelOneOpt,\labelTwoEquiX) \\
  + \costOne(\labelOneEqui,\labelTwoEqui)
  \leq \frac{2w_1+2w_2}{2w_1+w_2}
  + \costTwo(\labelOneEqui,\labelTwoEquiX)
  + \costOne(\labelOneOpt,\labelTwoEquiX)
\end{multline*}
In order to show $\costOne(\labelOneEqui,\labelTwoEqui) + \costTwo(\labelOneEqui,\labelTwoEqui) \leq \frac{2w_1+2w_2}{2w_1+w_2}$
it remains to prove that
\begin{equation}
  \frac{2w_1+2w_2}{2w_1+w_2} \costOne(\labelOneOpt,\labelTwoOpt)
  + \frac{2w_1+2w_2}{2w_1+w_2} \costTwo(\labelOneOpt,\labelTwoEquiX)
  - \costTwo(\labelOneEqui,\labelTwoEquiX)
  - \costOne(\labelOneOpt,\labelTwoEquiX)
  \label{eq_dual_proof_final}
\end{equation}
is nonnegative.
Unfortunately, cancellations only occur on the resource variable level.
Consider a resource $r \in R$ with $\labelOneOpt,\labelOneEqui,\labelTwoEquiX \in r$.
The sum of the coefficients of $\beta_r$ is at least
\begin{align*}
  &\frac{2w_1+2w_2}{2w_1+w_2} w_1 + \frac{2w_1+2w_2}{2w_1+w_2} w_2 + \frac{2w_1^2+2w_1w_2}{2w_1+w_2} - w_2 - w_1
  - w_1 - w_2 \\
  =\; &\frac{ (2w_1^2+2w_1w_2) + (2w_1w_2+2w_2^2) + (2w_1^2+2w_1w_2) - (4w_1^2 + 2w_1w_2) - (4w_1w_2 + 2w_2^2) }{ 2w_1+w_2 }
  = 0
\end{align*}
If $\labelOneOpt,\labelTwoEquiX \in r$ but $\labelOneEqui \notin r$, then we obtain for the sum of the coefficients of $\beta_r$ at least
\begin{align*}
  &\frac{2w_1+2w_2}{2w_1+w_2} w_1 + \frac{2w_1+2w_2}{2w_1+w_2} w_2 + \frac{2w_1^2+2w_1w_2}{2w_1+w_2} - w_2 - w_1 -  w_2 \\
  =\; & \frac{ (2w_1^2+2w_1w_2) + (2w_1w_2+2w_2^2) + (2w_1^2+2w_1w_2) - (4w_1w_2 + 2w_2^2) - (2w_1^2 + w_1w_2) }{2w_1+w_2} \\
  =\; &\frac{ 2w_1^2 + w_1w_2 }{2w_1+w_2} \geq 0
\end{align*}
If $\labelOneEqui,\labelTwoEquiX \in r$ but $\labelOneOpt \notin r$, then we obtain for the sum of the coefficients of $\beta_r$ at least
\begin{align*}
  &\frac{2w_1+2w_2}{2w_1+w_2} w_2 - w_2 - w_1
  = \frac{ (2w_1w_2+2w_2^2) - (2w_1w_2 + w_2^2) - (2w_1^2 + w_1w_2) }{2w_1+w_2} \\
  =\; &\frac{ -2w_1^2 - w_1w_2 + w_2^2 }{2w_1+w_2}
  = \frac{ (w_2 - 2w_1) \cdot (w_2 + w_1) }{2w_1+w_2}
\end{align*}
which is nonnegative due to $2w_1 \leq w_2$.
The sum of the coefficients of $\beta_r$ for all other $r \in R$ as well as of $\alpha_r$ for all $r \in R$ is easily seen to be nonnegative since every contribution of the last two terms in~\eqref{eq_dual_proof_final} is compensated by one of the first two terms because $\frac{2w_1 + 2w_2}{2w_1 + w_2} \geq 1$ holds.
This concludes the proof that the price of anarchy for uniformly weighted sequential $2$-player games with weights satisfying $2w_1 \leq w_2$ is equal to $1 + w_2 / (2w_1 + w_2)$.

The proof for the other types of congestion games follow exactly the same method.
However, some require higher-degree polynomials and are thus even more technical.

\section{Concluding remarks}
\label{sec_conclusion}

One of the unexpected findings of our paper is the fact that the worst-cases for simultaneous games are attained for players' weights that differ only slightly, and that sequential play reduces the price of anarchy irrespective of of the players' weights.
However, this discrepancy vanishes for increasing weight difference.
While for simultaneous games the symmetry with respect to players implies that the prices of anarchy for weight ratios $\lambda$ and $1/\lambda$ are equal, no such implication holds for sequential games.
However, \cref{fig_plot} shows that the same symmetry also holds for sequential games with proportional costs.
Surprisingly, a different symmetry holds for sequential games with uniform costs, namely equality for weight ratios $\lambda/2$ and $1/\lambda$.

As to methodology, observe that the algebraic expressions for the price of anarchy are already quite complicated for two players, especially for symmetric simultaneous games.
Hence, a similar analysis for the 3-player case seems to be out of reach.


\bibliographystyle{plain}
\bibliography{weighted-price-of-anarchy-two-players}

\end{document}